\documentclass[11pt, a4paper]{article}
\usepackage[T1]{fontenc}
\usepackage{amsfonts}
\usepackage{amsmath}
\usepackage{amssymb}
\usepackage{amsthm}
\usepackage{bbm}
\usepackage{bm}
\usepackage{mathrsfs}
\usepackage{verbatim}
\usepackage{setspace}
\usepackage{enumitem}
\theoremstyle{plain}
\newtheorem{Thm}{Theorem}[section]
\newtheorem{Prop}[Thm]{Proposition}

\newtheorem{Lemma}[Thm]{Lemma}
\newtheorem{Cor}[Thm]{Corollary}

\theoremstyle{definition}

\newtheorem{Remark}[Thm]{Remark}

\theoremstyle{remark}

\newtheorem*{Lit}{Literature}

\renewenvironment{thebibliography}[1]{%
\begin{oldthebibliography}{#1}%
\setlength{\baselineskip}{.9em}
\linespread{.9}
\small
\setlength{\parskip}{0ex}%
\setlength{\itemsep}{.1em}%
}%
{%
\end{oldthebibliography}%
}
\newcommand{\q}{\quad}

\newcommand{\eps}{\varepsilon}
\newcommand{\N}{\mathbb{N}}

\newcommand{\R}{\mathbb{R}}

\newcommand{\F}{\mathbb{F}}
\newcommand{\cF}{\mathcal{F}}

\newcommand{\cB}{\mathcal{B}}

\newcommand{\cR}{\mathcal{R}}
\newcommand{\cA}{\mathcal{A}}
\newcommand{\cH}{\mathcal{H}}

\newcommand{\cE}{\mathcal{E}}

\newcommand{\cS}{\mathcal{S}}

\newcommand{\sM}{\mathscr{M}}

\newcommand{\sY}{\mathscr{Y}}

\newcommand{\loc}{\text{loc}}

\newcommand{\br}[1]{\langle #1 \rangle}

\DeclareMathOperator{\esssup}{ess\, sup}
\DeclareMathOperator{\essinf}{ess\, inf}

\newcommand{\hX}{\widehat{X}}
\newcommand{\hY}{\widehat{Y}}
\newcommand{\tU}{\widetilde{U}}

\newcommand{\tc}{\tilde{c}}
\newcommand{\tx}{\tilde{x}}
\newcommand{\hc}{\hat{c}}

\newcommand{\hpi}{\hat{\pi}}
\newcommand{\hkappa}{\hat{\kappa}}

\newcommand{\tL}{\widetilde{L}}
\newcommand{\tM}{\widetilde{M}}
\newcommand{\tN}{\widetilde{N}}
\newcommand{\tX}{\widetilde{X}}

\newcommand{\tZ}{\widetilde{Z}}

\newcommand{\hvartheta}{\hat{\vartheta}}
\newcommand{\tvartheta}{\tilde{\vartheta}}

\newcommand{\dL}{\delta L}
\newcommand{\dM}{\delta M}
\newcommand{\dN}{\delta N}

\newcommand{\dZ}{\delta Z}

\newcommand{\sint}{\stackrel{\mbox{\tiny$\bullet$}}{}}

\numberwithin{equation}{section}

\begin{document}%

\title{Risk Aversion Asymptotics for Power Utility Maximization\\[1em] \small{Marcel Nutz\\
ETH Zurich, Department of Mathematics, 8092 Zurich, Switzerland\\ \texttt{marcel.nutz@math.ethz.ch} \\
 This Version: March 16, 2010.}}
\date{}
\maketitle \vspace{-1.5cm}

\begin{abstract}
We consider the economic problem of optimal consumption and investment with power utility. We study the optimal strategy as the relative risk aversion tends to infinity or to one. The convergence of the optimal consumption is obtained for general semimartingale models while the convergence of the optimal trading strategy is obtained for continuous models. The limits are related to exponential and  logarithmic utility. To derive these results, we combine approaches from optimal control, convex analysis and backward stochastic differential equations (BSDEs).
\end{abstract}

{\small


}\\

\noindent \emph{Acknowledgements.} Financial support by Swiss National Science Foundation Grant PDFM2-120424/1 is gratefully acknowledged. The author thanks Freddy Delbaen and Semyon Malamud for discussions and Martin Schweizer for comments on the draft.

\section{Introduction}

This paper considers the maximization of expected utility, a classical problem of mathematical finance. The agent obtains utility from the wealth he possesses at some given time horizon $T\in(0,\infty)$ and, in an alternative case, also from intermediate consumption before $T$.
More specifically, we study preferences given by power utility random fields for an agent who can invest in a financial market which is modeled by a general semimartingale. We defer the precise formulation to the next section to allow for a brief presentation of the contents and focus on the power utility function $U^{(p)}(x)=\tfrac{1}{p}x^p$, where $p\in(-\infty,0)\cup(0,1)$. Under standard assumptions, there exists for each $p$ an optimal trading and consumption strategy that maximizes the expected utility corresponding to $U^{(p)}$. Our main interest concerns the \emph{behavior of these strategies in the limits $p\to-\infty$ and $p\to0$}.

The relative risk aversion of $U^{(p)}$ tends to infinity for $p\to-\infty$. Hence economic intuition suggests that the
agent should become reluctant to take risks and, in the limit, not invest in the risky assets.
Our first main result confirms this intuition. More precisely, we prove in a general semimartingale model that the optimal consumption, expressed as a proportion of current wealth, converges pointwise to a deterministic function. This function corresponds to the consumption which would be optimal in the case where trading is not allowed. In the continuous semimartingale case, we show that the optimal trading strategy tends to zero in a local $L^2$-sense and that the corresponding wealth process converges in the semimartingale topology.

Our second result pertains to the same limit $p\to-\infty$ but concerns the problem without intermediate consumption. In the continuous case, we show that the optimal trading strategy scaled by $1-p$ converges to a strategy which is optimal for exponential utility. We provide economic intuition for this fact via a sequence of auxiliary power utility functions with shifted domains.

The limit $p\to 0$ is related to the logarithmic utility function. Our third main result is the convergence of the corresponding optimal consumption for the general semimartingale case, and the convergence of the trading strategy and the wealth process in the continuous case.

All these results are readily observed for special models where the optimal strategies can be calculated explicitly. While the corresponding economic intuition extends to general models, it is \emph{a priori} unclear how to go about proving the results. Indeed, the problem is to \emph{get our hands on the optimal controls}, which is a notorious question in stochastic optimal control.

Our main tool is the so-called opportunity process, a reduced form of the value process in the sense of dynamic programming.
We prove its convergence using control-theoretic arguments and convex analysis. On the one hand, this yields the convergence of the value function. On the other hand, we deduce the convergence of the optimal consumption, which is directly related to the opportunity process.
The optimal trading strategy is also linked to this process, by the so-called Bellman equation. We study the asymptotics of this backward stochastic differential equation (BSDE) to obtain the convergence of the strategy. This involves nonstandard arguments to deal with nonuniform quadratic growth in the driver and solutions that are not locally bounded.

To derive the results in the stated generality, it is important to \emph{combine} ideas from optimal control, convex analysis and BSDE theory rather than to rely on only one of these ingredients; and one may see the problem at hand as a \emph{model problem of control} in a semimartingale setting.

The paper is organized as follows. In the next section, we specify the optimization problem in detail. Section~\ref{se:mainResults} summarizes the main results on the risk aversion asymptotics of the optimal strategies and indicates connections to the literature. Section~\ref{se:tools} introduces the main tools, the opportunity process and the Bellman equation, and explains the general approach for the proofs. In Section~\ref{se:auxResults} we study the dependence of the opportunity process on $p$ and establish some related estimates. Sections~\ref{se:pto-inftyProofs} deals with the limit $p\to-\infty$; we prove the main results stated in Section~\ref{se:mainResults} and, in addition, the convergence of the opportunity process and the solution to the dual problem (in the sense of convex duality). Similarly, Section~\ref{se:pto0Proofs} contains the proof of the main theorem for $p\to0$ and additional refinements.

\section{Preliminaries}

The following notation is used. If $x,y\in\R$ are reals, $x\wedge y=\min\{x,y\}$ and $x\vee y=\max\{x,y\}$.
We use $1/0:=\infty$ where necessary. If $z\in\R^d$ is a $d$-dimensional vector, $z^i$ is its $i$th coordinate, $z^\top$ its transpose, and
$|z|=(z^\top z)^{1/2}$ the Euclidean norm. If $X$ is an $\R^d$-valued semimartingale
and $\pi$ is an $\R^d$-valued predictable integrand, the vector stochastic integral, denoted by $\int \pi\,dX$ or $\pi\sint X$, is a scalar semimartingale with initial value zero.
Relations between measurable functions
hold almost everywhere unless otherwise mentioned. Dellacherie and Meyer~\cite{DellacherieMeyer.82} and Jacod and Shiryaev~\cite{JacodShiryaev.03} are references for unexplained notions from stochastic calculus.

\subsection{The Optimization Problem}

We consider a fixed time horizon $T\in(0,\infty)$ and a filtered probability space $(\Omega,\cF,\F=(\cF_t)_{t\in[0,T]},P)$  satisfying the usual assumptions of right-continuity and completeness, as well as $\cF_0=\{\emptyset,\Omega\}$ $P$-a.s.\ Let $R$ be an $\R^d$-valued c\`adl\`ag semimartingale with $R_0=0$. Its components are interpreted as the returns of $d$ risky assets and
the stochastic exponential $S=(\cE(R^1),\dots,\cE(R^d))$ represents their prices.
Let $\sM$ be the set of equivalent $\sigma$-martingale measures for $S$. We assume
\begin{equation}\label{eq:ELMMexists}
  \sM\neq\emptyset,
\end{equation}
so that arbitrage is excluded in the sense of the NFLVR condition (see Delbaen and Schachermayer~\cite{DelbaenSchachermayer.98}).
Our agent also has a bank account at his disposal. As usual in mathematical finance, the interest rate is assumed to be zero.

The agent is endowed with a deterministic initial capital $x_0>0$. A \emph{trading strategy} is a predictable $R$-integrable $\R^d$-valued process $\pi$, where $\pi^i$ is interpreted as the fraction of the current wealth (or the portfolio proportion) invested in the $i$th risky asset. A \emph{consumption rate} is an optional process $c\geq0$ such that $\int_0^Tc_t\,dt<\infty$ $P$-a.s.
We want to consider two cases simultaneously: Either consumption occurs only at the terminal time $T$ (utility from ``terminal wealth'' only); or there is intermediate and a bulk consumption at the time horizon. To unify the notation, we define the measure $\mu$ on $[0,T]$,
 \[
   \mu(dt):=
  \begin{cases}
    0 & \text{in the case without intermediate consumption}, \\
    dt & \text{in the case with intermediate consumption}.
  \end{cases}
 \]
 Moreover, let $\mu^\circ:=\mu + \delta_{\{T\}}$, where $\delta_{\{T\}}$ is the unit Dirac measure at $T$.
The \emph{wealth process} $X(\pi,c)$ of a pair $(\pi,c)$ is defined by the linear equation
\begin{equation*}%
   X_t(\pi,c)=x_0+\int_0^t  X_{s-}(\pi,c) \pi_s\,dR_s-\int_0^t c_s\,\mu(ds),\quad 0\leq t\leq T.
\end{equation*}
The set of \emph{admissible} trading and consumption pairs is
\[
  \cA(x_0)=\big\{(\pi,c):\, X(\pi,c)>0 \mbox{ and } c_T=X_T(\pi,c)\big\}.
\]
The convention $c_T=X_T(\pi,c)$ is merely for notational convenience and means that all the remaining wealth is consumed at time $T$.
We fix the initial capital $x_0$ and usually write $\cA$ for $\cA(x_0)$.
Moreover, $c\in\cA$ indicates that there exists $\pi$ such that $(\pi,c)\in\cA$; an analogous convention is used for similar expressions.

It will be convenient to parametrize the consumption strategies as fractions of the wealth.
Let $(\pi,c)\in\cA$ and let $X=X(\pi,c)$ be the corresponding wealth process. Then
\begin{equation*}%
  \kappa:=\frac{c}{X}
\end{equation*}
is called the \emph{propensity to consume} corresponding to $(\pi,c)$.
In general, a propensity to consume is an optional process $\kappa\geq0$ such that $\int_0^T \kappa_s\,ds<\infty$ $P$-a.s.\ and $\kappa_T=1$. The parametrizations by $c$ and by $\kappa$ are equivalent (see Nutz~\cite[Remark~2.1]{Nutz.09a}) and
we abuse the notation by identifying $c$ and $\kappa$ when $\pi$ is given. Note that the wealth process can be expressed as
\begin{equation}\label{eq:wealthExponential}
   X(\pi,\kappa)=x_0\cE\big(\pi\sint R - \kappa \sint \mu\big).
\end{equation}

The preferences of the agent are modeled by a random utility function with constant relative risk aversion. More precisely, let $D$ be a c\`adl\`ag adapted positive process and fix $p\in (-\infty,0)\cup(0,1)$. We define the utility random field
\begin{equation}\label{eq:utilityDef}
  U_t(x):=U^{(p)}_t(x):=D_t\tfrac{1}{p}x^p,\q x\in (0,\infty),\; t\in [0,T],
\end{equation}
where we assume that there are constants $0<k_1\leq k_2<\infty $ such that
\begin{equation}\label{eq:BoundsR}
  k_1\leq D_t\leq k_2,\q 0\leq t\leq T.
\end{equation}
The process $D$ is taken to be independent of $p$; interpretations are discussed in~\cite[Remark~2.2]{Nutz.09a}.
The parameter $p$ in $U^{(p)}$ will sometimes be suppressed in the
notation and made explicit when we want to recall the dependence. The same applies to other quantities in this paper.

The constant $1-p>0$ is called the \emph{relative risk aversion} of $U$.
The \emph{expected utility} corresponding to a consumption rate $c\in\cA$ is given by
$E\big[\int_0^T U_t(c_t)\,\mu^\circ(dt)\big]$, which is
either $E[U_T(c_T)]$ or $E[\int_0^T U_t(c_t)\,dt+U_T(c_T)]$.
We will always assume that the optimization problem is nondegenerate, i.e.,
\begin{equation}\label{eq:PrimalProblemFinite}
  u_p(x_0):=\sup_{c\in\cA(x_0)}E\Big[\int_0^T U^{(p)}_t(c_t)\,\mu^\circ(dt)\Big]<\infty.
\end{equation}
This condition depends on the choice of $p$, but not on $x_0$.
Note that $u_{p_0}(x_0)<\infty$ implies $u_p(x_0)<\infty$ for any $p<p_0$; and for $p<0$ the condition~\eqref{eq:PrimalProblemFinite} is void since then $U^{(p)}<0$.
A strategy $(\pi,c)\in\cA(x_0)$
is \emph{optimal} if $E\big[\int_0^T U_t(c_t)\,\mu^\circ(dt)\big]=u(x_0)$. Note that $U_t$  is irrelevant for $t<T$ when there is no intermediate consumption.
We recall the following existence result.

\begin{Prop}[Karatzas and \v{Z}itkovi\'c~\cite{KaratzasZitkovic.03}]\label{pr:ExistenceKZ}
  For each $p$, if $u_{p}(x_0)<\infty$, there exists an optimal strategy $(\hpi,\hc)\in\cA$.
  The corresponding wealth process $\hX=X(\hpi,\hc)$ is unique. The consumption rate $\hc$ can be chosen
  to be c\`adl\`ag and is unique $P\otimes\mu^\circ$-a.e.
\end{Prop}

In the sequel, $\hat{c}$ denotes this c\`adl\`ag version, $\hX=X(\hpi,\hc)$ is the optimal wealth process and
$\hkappa=\hc/\hX$ is the optimal propensity to consume.

\subsection{Decompositions and Spaces of Processes}\label{se:decomps}
In some of the statements, we will assume that the price process $S$ (or equivalently $R$) is continuous. In this case,
it follows from~\eqref{eq:ELMMexists} and Schweizer~\cite{Schweizer.95b} that $R$ satisfies the \emph{structure condition},
i.e.,
\begin{equation}\label{eq:StructureContForR}
  R=M +\int d\br{M}\lambda,
\end{equation}
where $M$ is a continuous local martingale with $M_0=0$ and $\lambda\in L^2_{loc}(M)$.

Let $\xi$ be a scalar special semimartingale, i.e., there exists a (unique) canonical decomposition
$\xi=\xi_0 + M^\xi+ A^\xi$, where $\xi_0\in\R$, $M^\xi$ is a local martingale, $A^\xi$ is predictable of finite variation, and
$M_0^\xi=A_0^\xi=0$. As $M$ is continuous,
$M^\xi$ has a Kunita-Watanabe (KW) decomposition with respect to~$M$,
\begin{equation}\label{eq:GKWdecompL}
  \xi=\xi_0 +Z^\xi\sint M + N^\xi+ A^\xi,
\end{equation}
where $[M^i,N^\xi]=0$ for $1\leq i\leq d$ and $Z^\xi\in L^2_{loc}(M)$; see Ansel and Stricker~\cite[\emph{cas}~3]{AnselStricker.03}.
Analogous notation will be used for other special semimartingales and, with a slight abuse of terminology, we will refer to~\eqref{eq:GKWdecompL} as the KW decomposition of $\xi$.

Let $\cS$ be the space of all c\`adl\`ag $P$-semimartingales and $r\in[1,\infty)$. If $X\in\cS$ has the canonical decomposition $X=X_0+M^X+A^X$, we define
\[
  \|X\|_{\cH^r}:=|X_0|+ \big\|\textstyle{\int}_0^T |dA^X|\big\|_{L^r}+ \big\|[M^X]_T^{1/2}\big\|_{L^r}.
\]
In particular, we will often use that $\|N\|^2_{\cH^2}=E\big[[N]_T\big]$ for a local martingale $N$ with $N_0=0$.
If $X$ is a non-special semimartingale, $\|X\|_{\cH^r}:=\infty$. We can now define
$\cH^r:=\{X\in\cS:\, \|X\|_{\cH^r}<\infty\}$. The same space is sometimes denoted by $\cS^r$ in the literature; moreover, there are many equivalent definitions for $\cH^r$ (see~\cite[VII.98]{DellacherieMeyer.82}). The localized spaces $\cH^r_{loc}$ are defined in the usual way. In particular, if $X,X^n\in\cS$ we say that $X^n\to X$ in $\cH^r_{loc}$ if there exists a localizing sequence of stopping times $(\tau_m)_{m\geq1}$
such that
$\lim_n\|(X^n-X)^{\tau_m}\|_{\cH^r}=0$ for all $m$. The localizing sequence may depend on the sequence $(X^n)$, causing this convergence to be non-metrizable.
On $\cS$, the \'Emery distance is defined by
\[
  d(X,Y):=|X_0-Y_0|+ \sup_{|H|\leq 1} E\bigg[\sup_{t\in[0,T]} 1\wedge |H\sint (X-Y)_t| \bigg],
\]
where the supremum is taken over all predictable processes bounded by one in absolute value.
This complete metric induces on $\cS$ the \emph{semimartingale topology} (cf.~\'Emery~\cite{Emery.79}).

An optional process $X$ satisfies a certain property \emph{prelocally} if there exists a localizing sequence
of stopping times $\tau_m$ such that
$X^{\tau_m-}:=X1_{[0,\tau_m)}+X_{\tau_m-}1_{[\tau_m,T]}$ satisfies this property for each $m$.
When $X$ is continuous, prelocal simply means local.

\begin{Prop}[\cite{Emery.79}]\label{pr:SMconvPrelocHp}
  Let $X,X^n\in \cS$ and $r\in[1,\infty)$. Then $X^n\to X$ in the semimartingale topology if and only if
  every subsequence of $(X^n)$ has a subsequence which converges to $X$ prelocally in $\cH^r$.
\end{Prop}

We denote by $BMO$ the space of martingales $N$ with $N_0=0$ satisfying
\[
   \|N\|_{BMO}^2:=\Big\|\sup_\tau E\big[[N]_T-[N]_{\tau-}\big|\cF_{\tau}\big]\Big\|_{L^\infty}<\infty,
\]
where $\tau$ ranges over all stopping times (more precisely, this is the $BMO_2$-norm).
There exists a similar notion for semimartingales: let $\cH^\omega$ be the subspace of $\cH^1$
consisting of all special semimartingales $X$ with $X_0=0$ and
\[
   \|X\|_{\cH^\omega}^2:=\Big\|\sup_\tau E\Big[\big([M^X]_T-[M^X]_{\tau-}\big)^{1/2} + \textstyle{\int}_{\tau-}^T |dA^X| \,\Big|\cF_{\tau}\Big]\Big\|_{L^\infty}<\infty.
\]

Finally, let $\cR^r$ be the space of scalar adapted processes which are right-continuous and such that
\[
  \|X\|_{\cR^r}:=\Big\|\sup_{0\leq t\leq T} |X_t|\Big\|_{L^r}<\infty.
\]
With a mild abuse of notation, we will use the same norm also for left-continuous processes.

\section{Main Results}\label{se:mainResults}

In this section we present the main results about the limits of the optimal strategies.
To state an assumption in the results, we first have to introduce the \emph{opportunity process $L(p)$}; this is a reduced form of the value process in the language of dynamic programming.
Fix $p$ such that $u_p(x_0)<\infty$. Using the scaling properties of our utility function, we can show that there exists a unique c\`adl\`ag semimartingale $L(p)$ such that
\begin{equation}\label{eq:OppProcIndep}
    L_t(p)\,\tfrac{1}{p}\big(X_t(\pi,c)\big)^p= \mathop{\esssup}_{\tilde{c}\in\cA(\pi,c,t)} E\Big[\int_t^T U_s(\tc_s)\,\mu^\circ(ds)\Big|\cF_t\Big],\quad 0\leq t\leq T
\end{equation}
for all $(\pi,c)\in\cA$, where
$\cA(\pi,c,t):=\big\{(\tilde{\pi},\tilde{c})\in\cA:\, (\tilde{\pi},\tilde{c})=(\pi,c)\mbox{ on }[0,t]\big\}$.
While we refer to~\cite[Proposition~3.1]{Nutz.09a} for the proof, we shall have more to say about $L(p)$ later since it will be an important tool in our analysis.

We can now proceed to state the main results. The proofs are postponed to Sections~\ref{se:pto-inftyProofs} and~\ref{se:pto0Proofs}. Those sections also contain statements about the convergence of the opportunity processes and the solutions to the dual problems, as well as some refinements of the results below.

\subsection{The Limit $p\to-\infty$}\label{se:mainResInfty}

The relative risk aversion $1-p$ of $U^{(p)}$ increases to infinity as $p\to-\infty$. Therefore we expect that in the limit, the agent does not invest at all. In that situation the optimal propensity to consume is $\kappa_t=(1+T-t)^{-1}$ since this corresponds to a constant consumption rate. Our first result shows that this coincides with the limit of the $U^{(p)}$-optimal propensities to consume.

\begin{Thm}\label{th:Limit-inftyEconomics}
 The following convergences hold as $p\to-\infty$.
  \begin{enumerate}[topsep=3pt, partopsep=0pt, itemsep=1pt,parsep=2pt]
    \item Let $t\in[0,T]$. In the case with intermediate consumption,
      \[
        \hkappa_t(p)\to \frac{1}{1+T-t} \quad P\mbox{-a.s.}
      \]
      If $\F$ is continuous, the convergence is uniform in $t$, $P$-a.s.; and holds also in $\cR^r_{loc}$ for all $r\in [1,\infty)$.
    \item If $S$ is continuous and $L(p)$ is continuous for all $p<0$, then
     \[
       \hpi(p)\to 0 \mbox{ in }L^2_{loc}(M)
     \]
     and $\hX(p)\to x_0\exp\big(-\int_0^{\cdot} \frac{\mu(ds)}{1+T-s}\big)$ in the semimartingale topology.
  \end{enumerate}
\end{Thm}

The continuity assumptions in~(ii) are always satisfied if the filtration $\F$ is generated by a Brownian motion; see also Remark~\ref{rk:Lcontinuity}.

\begin{Lit}
  We are not aware of a similar result in the continuous-time literature, with the exception that when the strategies can be calculated explicitly, the convergences mentioned in this section are often straightforward to obtain. E.g., Grasselli~\cite{Grasselli.03} carries out such a construction in a complete market model. There are also related systematic results.
  Carassus and R\'asonyi~\cite{CarassusRasonyi.06} and Grandits and Summer~\cite{GranditsSummer.07} study convergence to the superreplication problem for increasing (absolute) risk aversion of general utility functions in discrete models. Note that superreplicating the contingent claim $B\equiv0$ corresponds to not trading at all. For the maximization of exponential utility $-\exp(-\alpha x)$ without claim, the optimal strategy is proportional to the inverse of the absolute risk aversion $\alpha$ and hence trivially converges to zero in the limit $\alpha\to\infty$. The case with claim is also studied. See, e.g., Mania and Schweizer~\cite{ManiaSchweizer.05} for a continuous model, and Becherer~\cite{Becherer.06} for a related result. The references given here and later in this section do not consider intermediate consumption.
\end{Lit}

We continue with our second main result, which concerns only the case without intermediate consumption.
We first introduce in detail the exponential hedging problem already mentioned above. Let $B\in L^\infty(\cF_T)$ be a contingent claim.  Then the aim is to maximize
the expected exponential utility (here with $\alpha=1$) of the terminal wealth including the claim,

\begin{equation}\label{eq:exponentialProblem}
  \max_{\vartheta\in\Theta} E\big[-\exp\big(B - x_0-(\vartheta\sint R)_T\big)\big],
\end{equation}
where $\vartheta$ is the trading strategy parametrized by the \emph{monetary amounts} invested in the assets
(setting $\overline{\vartheta}^i:=1_{\{S^i_-\neq0\}} \vartheta^i/S^i_-$ yields $\overline{\vartheta} \sint S=\vartheta\sint R$ and corresponds to the more customary \emph{number} of shares of the assets).

To describe the set $\Theta$, we define the \emph{entropy} of $Q\in\sM$ relative to $P$ by
\[
 H(Q|P):=E\Big[\frac{dQ}{dP}\log\Big(\frac{dQ}{dP}\Big)\Big]=E^Q\Big[\log\Big(\frac{dQ}{dP}\Big)\Big]
\]
and let $\sM^{ent}=\big\{Q\in\sM:\,H(Q|P)<\infty\big\}$. We assume in the following that
\begin{equation}\label{eq:finiteEntropy}
 \sM^{ent}\neq\emptyset.
\end{equation}
Now $\Theta:=\big\{\vartheta\in L(R): \vartheta\sint R \mbox{ is a $Q$-supermartingale for all }Q\in\sM^{ent}\big\}$ is the class
of admissible strategies for~\eqref{eq:exponentialProblem}. If $S$ is locally bounded,
there exists an optimal
strategy $\hvartheta\in\Theta$ for~\eqref{eq:exponentialProblem} by Kabanov and Stricker~\cite[Theorem~2.1]{KabanovStricker.02}.
(See Biagini and Fritelli~\cite{BiaginiFrittelli.05, BiaginiFrittelli.07} for the unbounded case.)

As there is no intermediate consumption, the process $D$ in~\eqref{eq:utilityDef} reduces to a random variable $D_T\in L^\infty(\cF_T)$. If we choose
\begin{equation}\label{eq:ClaimIsDiscounting}
  B:=\log(D_T),
\end{equation}
we have the following result.

\begin{Thm}\label{th:strategyConvExponentialMainRes}
  Let $S$ be continuous and assume that $L(p)$ is continuous for all $p<0$. Under~\eqref{eq:finiteEntropy} and~\eqref{eq:ClaimIsDiscounting},
  \[
    (1-p)\, \hpi(p) \to \hvartheta\quad \mbox{ in } L^2_{loc}(M).
  \]
  Here $\hpi(p)$ is in the fractions of wealth parametrization, while $\hvartheta$ denotes the monetary amounts invested for the exponential utility.
\end{Thm}

As this convergence may seem surprising at first glance, we give the following heuristics.

\begin{Remark}\label{rk:heuristicsExpConv}
  Assume $B=\log(D_T)=0$ for simplicity. The preferences induced by $U^{(p)}(x)=\tfrac{1}{p}x^p$ on $\R_+$ are not directly comparable to the ones given by the exponential utility, which are defined on $\R$.
  We
  consider the shifted power utility functions
  \[
    \tU^{(p)}(x):= U^{(p)}\big(x+ 1-p \big),\quad x\in (p-1,\infty).
  \]
  Then $\tU^{(p)}$ again has relative risk aversion $1-p>0$ and its domain of definition increases to $\R$ as $p\to-\infty$.
  Moreover,
  \begin{equation}\label{eq:ModifiedPowerToExp}
    (1-p)^{1-p}\; \tU^{(p)}(x) = \tfrac{1-p}{p}\Big(\frac{x}{1-p}+1\Big)^p \to\; -e^{-x},\quad p\to -\infty,
  \end{equation}
  and the multiplicative constant does not affect the preferences.

  Let the agent with utility function $\tU^{(p)}$ be endowed with some initial capital $x^*_0\in\R$ independent of
  $p$. (If $x_0^*<0$, we consider only values of $p$ such that $p-1<x_0^*$.) The change of variables
  $x=\tx + 1-p$
  yields $U^{(p)}(x)=\tU^{(p)}(\tx)$.  Hence the corresponding optimal wealth processes $\hX(p)$ and $\tX(p)$ are related by
  $\tX(p)=\hX(p)-1+p$ if we choose the initial capital $x_0:=x_0^*+1-p>0$ for the agent with $U^{(p)}$.
  We conclude
  \[
    d\tX(p) = d\hX(p)=\hX(p)\hpi(p)\,dR = \big(\tX(p)+1-p\big)\hpi(p)\,dR,
  \]
  i.e., the optimal monetary investment $\tilde{\vartheta}(p)$ for $\tU^{(p)}$ is given by
  \[
    \tilde{\vartheta}(p)=\big(\tX(p)+1-p\big)\hpi(p).
  \]
  In view of~\eqref{eq:ModifiedPowerToExp}, it is reasonable that $\tilde{\vartheta}(p)$ should converge to $\hvartheta$, the optimal monetary investment for the exponential utility.
  We recall that $\hpi(p)$ (in fractions of wealth) does not depend on $x_0$ and converges to zero
  under the conditions of Theorem~\ref{th:Limit-inftyEconomics}. Thus, loosely speaking, $\tX(p)\hpi(p)\approx 0$ for $-p$ large, and hence
  \[
    \tilde{\vartheta}(p)\approx (1-p) \hpi(p).
  \]
  More precisely, one can show that $\lim_{p\to-\infty} \big(\tX(p)\hpi(p)\big) \sint R =0$ in the semimartingale topology, using arguments as in Appendix~\ref{se:convExponentials}.
\end{Remark}

\begin{Lit}
  To the best of our knowledge, the statement of Theorem~\ref{th:strategyConvExponentialMainRes} is new in the systematic literature. However, there are known results on the dual side for the case $B=0$. The problem dual to exponential utility maximization is the minimization
  of $H(Q|P)$ over to $\sM^{ent}$ and the optimal $Q^E\in \sM^{ent}$ is called minimal entropy martingale measure.
  Under additional assumptions on the model, the solution $\hY(p)$ of the dual problem for power utility~\eqref{eq:dualProblem} introduced below is a martingale
  and then the measure $Q^q$ defined by $dQ^q/dP=\hY_T(p)/\hY_0(p)$ is called $q$-optimal martingale measure, where $q<1$ is conjugate to $p$. This measure can be defined also for $q>1$, in which case it is not connected to power utility.
  The convergence of $Q^q$ to $Q^E$ for $q\to1+$ was proved by Grandits and Rheinl\"ander~\cite{GranditsRheinlander.02}
  for continuous semimartingale models satisfying a reverse H\"older inequality. Under the additional assumption
  that $\F$ is continuous, the convergence of $Q^q$ to $Q^E$ for $q\to1$ and more generally the continuity of $q\mapsto Q^q$ for $q\geq0$ were obtained by Mania and Tevzadze~\cite{ManiaTevzadze.03} (see also Santacroce~\cite{Santacroce.05}) using BSDE convergence together with $BMO$ arguments. The latter are possible due to the reverse H\"older inequality; an assumption which is not present in our results.
\end{Lit}

\subsection{The Limit $p\to0$}

As $p$ tends to zero, the relative risk aversion of the power utility tends to $1$, which corresponds to the
utility function $\log(x)$. Hence we consider
\[
  u_{\log}(x_0):=\sup_{c\in\cA(x_0)}E\Big[\int_0^T \log(c_t)\,\mu^\circ(dt)\Big];
\]
here integrals are set to $-\infty$ if they are not well defined in $\overline{\R}$.
A $\log$-utility agent exhibits a very special (``myopic'') behavior,
which allows for an explicit solution of the utility maximization problem (cf.~Goll and Kallsen~\cite{GollKallsen.00, GollKallsen.03}). If in particular $S$ is continuous, the $\log$-optimal strategy is
\[
  \pi_t=\lambda_t,\q \kappa_t=\frac{1}{1+T-t}
\]
by~\cite[Theorem 3.1]{GollKallsen.00}, where $\lambda$ is defined by~\eqref{eq:StructureContForR}.
Our result below shows that the optimal strategy for power utility with $D\equiv 1$ converges to the
$\log$-optimal one as $p\to0$.
In general, the randomness of $D$ is an additional source of risk and will cause an excess hedging demand.
Consider the bounded semimartingale
\[
  \eta_t:=E\Big[\int_t^T D_s\,\mu^\circ(ds)\Big|\cF_t\Big].
\]
If $S$ is continuous,
$\eta = \eta_0  + Z^\eta \sint M + N^\eta+ A^\eta$
denotes the Kunita-Watanabe decomposition of $\eta$ with respect to $M$ and
the standard case $D\equiv1$ corresponds to $\eta_t=\mu^\circ[t,T]$ and $Z^\eta=0$.

\newpage

\begin{Thm}\label{th:Limit0Economics}
 Assume $u_{p_0}(x_0)<\infty$ for some $p_0\in(0,1)$. As $p\to0$,
  \begin{enumerate}[topsep=3pt, partopsep=0pt, itemsep=1pt,parsep=2pt]
    \item in the case with intermediate consumption,
      \[
        \hkappa_t(p)\to \frac{D_t}{\eta_t}\quad \mbox{uniformly in $t$, $P$-a.s.}
      \]
    \item if $S$ is continuous,
     \[
       \hpi(p)\to \lambda + \frac{Z^\eta}{\eta_-} \q\mbox{ in }L^2_{loc}(M)
     \]
     and the corresponding wealth processes converge in the semimartingale topology.
  \end{enumerate}
\end{Thm}

\begin{Remark}\label{rk:Limit0Economics}
  If we consider the limit $p\to 0-$, we need not \emph{a priori} assume that $u_{p_0}(x_0)<\infty$ for some $p_0>0$.
  Without that condition, the assertions of Theorem~\ref{th:Limit0Economics} remain valid if (i) is replaced by the weaker statement that $\lim_{p\to0-} \hkappa_t(p)\to D_t/\eta_t$ $P$-a.s.\ for all $t$. If $\F$ is continuous, (i) remains valid
  without changes. In particular, these convergences hold even if $u_{\log}(x_0)=\infty$.
\end{Remark}

\begin{Lit}
  In the following discussion we assume $D\equiv1$ for simplicity. It is part of the folklore that the $\log$-optimal strategy can be obtained from $\hpi(p)$ by formally setting $p=0$. Initiated by Jouini and Napp~\cite{JouiniNapp.04}, a recent branch of the literature studies the stability of the utility maximization problem under perturbations of the utility function (with respect to pointwise convergence) and other ingredients of the problem. To the best of our knowledge, intermediate consumption was not considered so far and the results for continuous time concern continuous semimartingale models.

  We note that $\log(x)=\lim_{p\to0} (U^{(p)}(x)-p^{-1})$ and here the additive constant does not influence the optimal strategy, i.e., we have pointwise convergence of utility functions ``equivalent'' to $U^{(p)}$. Now Larsen~\cite[Theorem~2.2]{Larsen.09} implies that the optimal terminal wealth $\hX_T$ for $U^{(p)}$ converges in probability to the $\log$-optimal one and that the value functions at time zero converge pointwise (in the continuous case without consumption).
  We use the specific form of our utility functions and obtain a stronger result.
  Finally, we can mention that on the dual side and for $p\to0-$, the convergence is related to the continuity of $q$-optimal measures as mentioned after Remark~\ref{rk:heuristicsExpConv}.
\end{Lit}

For general $D$ and $p$, it seems difficult to determine the precise influence of $D$ on the optimal trading strategy $\hpi(p)$. We can read Theorem~\ref{th:Limit0Economics}(ii) as a partial result on the excess hedging demand $\hpi(p)-\hpi(p,1)$ due to $D$; here $\hpi(p,1)$ denotes the optimal strategy for the case $D\equiv1$.

\begin{Cor}
  Suppose that the conditions of Theorem~\ref{th:Limit0Economics}(ii) hold.
  Then $\hpi(p)-\hpi(p,1)\to Z^\eta / \eta_-$ in $L^2_{loc}(M)$; i.e.,
  the asymptotic excess hedging demand due to $D$ is given by $Z^\eta / \eta_-$.
\end{Cor}
The stability theory mentioned above considers also perturbations of the probability measure $P$
(see Kardaras and {\v Z}itkovi\'c~\cite{KardarasZitkovic.09}) and our
corollary can be related as follows. In the special case when $D$ is a martingale, $U^{(p)}$ under $P$ corresponds to the standard power utility function optimized under the measure $d\widetilde{P}=(D_T/D_0)\,dP$ (see~\cite[Remark~2.2]{Nutz.09a}).
The excess hedging demand due to $D$ then represents the influence of the ``subjective beliefs'' $\widetilde{P}$.

\section{Tools and Ideas for the Proofs}\label{se:tools}

In this section we introduce our main tools and then present the basic ideas how
to apply them for the proofs of the theorems.

\subsection{Opportunity Processes}

We fix $p$ and assume $u_p(x_0)<\infty$ throughout this section. We first discuss the properties of the (primal) opportunity process
$L=L(p)$ as introduced in~\eqref{eq:OppProcIndep}. Directly from that equation we have that $L_T=D_T$ and
that $u_p(x_0)=L_0\tfrac{1}{p}x_0^p$ is the value function from~\eqref{eq:PrimalProblemFinite}.
Moreover, $L$ has the following properties by~\cite[Lemma~3.5]{Nutz.09a} in view of~\eqref{eq:BoundsR}.

\begin{Lemma}\label{le:BoundsForL} The opportunity process satisfies $L,L_->0$.
  \begin{enumerate}[topsep=3pt, partopsep=0pt, itemsep=1pt,parsep=2pt]
    \item If $p\in (0,1)$, $L$ is a supermartingale satisfying
    \[
      L_t\geq \big(\mu^\circ[t,T]\big)^{-p}\,E\Big[\int_t^T D_s\, \mu^\circ(ds)\Big|\cF_t\Big]\geq k_1.
    \]
    \item If $p<0$, $L$ is a bounded semimartingale satisfying
    \[
      0<L_t\leq \big(\mu^\circ[t,T]\big)^{-p}\,E\Big[\int_t^T D_s\, \mu^\circ(ds)\Big|\cF_t\Big] \leq k_2 \big(\mu^\circ[t,T]\big)^{1-p}.
    \]
    If in addition there is no intermediate consumption, then $L$ is a submartingale.
  \end{enumerate}
\end{Lemma}
In particular, $L$ is always a special semimartingale.
We denote by
\begin{equation}\label{eq:betaAndq}
  \beta:=\frac{1}{1-p}>0,\quad  q:=\frac{p}{p-1}\in(-\infty,0)\cup (0,1)
\end{equation}
the relative risk tolerance and the exponent conjugate to $p$, respectively. These constants are of course redundant given $p$, but turn out to simplify the notation.

In the case with intermediate consumption, the opportunity process and the optimal consumption are related by
\begin{equation}\label{eq:consumptionFeedback}
  \hc_t=\Big(\frac{D_t}{L_t}\Big)^\beta\hX_t \q\q\mbox{and hence}\q\q \hkappa_t=\Big(\frac{D_t}{L_t}\Big)^\beta
\end{equation}
according to~\cite[Theorem~5.1]{Nutz.09a}. %
Next, we introduce the convex-dual analogue of $L$; cf.~\cite[\S4]{Nutz.09a} for the following
notions and results.
The \emph{dual problem} is
\begin{equation}\label{eq:dualProblem}
  \inf_{Y\in\sY} E\Big[\int_0^T U_t^*(Y_t)\,\mu^\circ(dt)\Big],
\end{equation}
where $U_t^*(y)=\sup_{x>0} \big\{U_t(x)-xy\big\}=-\tfrac{1}{q}y^{q}D_t^{\beta}$ is the conjugate of $U_t$.
Only three properties of the domain $\sY=\sY(p)$ are relevant for us. First, each element $Y\in \sY$ is a positive c\`adl\`ag supermartingale. Second,
the set $\sY$ depends on $p$ only by a normalization: with the constant $y_0(p):=L_0(p)x_0^{p-1}$, the set $\sY':=y_0(p)^{-1} \sY(p)$ does not depend on $p$. As the elements of $\sY$ will occur only in terms of certain fractions, the constant
plays no role. Third, the $P$-density process of any $Q\in\sM$ is contained in $\sY$ (modulo scaling).

The \emph{dual opportunity process} $L^*$ is the analogue of $L$ for the dual problem and can be defined by
\begin{equation}\label{eq:DualOppProcAltDef}
L^*_t :=
  \begin{cases}
    \mathop{\esssup}_{Y\in\sY}& \hspace{-.8em} E\Big[\int_t^T D_s^\beta (Y_s/Y_t)^q\,\mu^\circ(ds)\Big|\cF_t\Big]
    \;\text{ if }p<0\phantom{\bigg|},\\
    \mathop{\essinf}_{Y\in\sY}& \hspace{-.8em} E\Big[\int_t^T D_s^\beta (Y_s/Y_t)^q\,\mu^\circ(ds)\Big|\cF_t\Big]
    \;\text{ if }p\in(0,1).
  \end{cases}
\end{equation}
Here the extremum is attained at the minimizer $Y\in\sY$ for~\eqref{eq:dualProblem}, which we denote by $\hY=\hY(p)$. %
Finally, we shall use that the primal and the dual opportunity process are related by the power
\begin{equation}\label{eq:DualAndPrimalOppProc}
  L^*=L^\beta.
\end{equation}

\subsection{Bellman BSDE}
We continue with a fixed $p$ such that $u_p(x_0)<\infty$. We recall the Bellman equation, which in the present paper will be used only for continuous $S$.
In this case, recall~\eqref{eq:StructureContForR} and let $L=L_0+Z^L\sint M + N^L+ A^L$ be the
KW decomposition of $L$ with respect to $M$.
Then the triplet $(L,Z^L,N^L)$ satisfies the Bellman BSDE
\begin{align}\label{eq:BellmanBSDEforL}
  d L_t & = \frac{q}{2}\, L_{t-}\Big(\lambda_t+\frac{Z^L_t}{L_{t-}}\Big)^\top\,d\br{M}_t\, \Big(\lambda_t+\frac{Z^L_t}{L_{t-}}\Big)\;
              -p U^*_t( L_{t-})\,\mu(dt)\nonumber\\
       &\phantom{=}\; + Z^L_t\,dM_t + dN^L_t; \\
   L_T  & = D_T. \nonumber
\end{align}
Put differently, the finite variation part of $L$ satisfies
\begin{equation}\label{eq:FVpartofL}
    A^L_t
    =\frac{q}{2}\int_0^t \,L_{s-}\Big(\lambda_s+\frac{Z^L_s}{L_{s-}}\Big)^\top\,d\br{M}_s\, \Big(\lambda_s+\frac{Z^L_s}{L_{s-}}\Big)\; -p \int_0^t U^*_s(L_{s-})\,\mu(ds).
\end{equation}
Here $U^*$ is defined as in~\eqref{eq:dualProblem}. Moreover, the optimal trading strategy $\hpi$ can be described by
\begin{equation}\label{eq:optStrategy}
  \hpi_t=\beta\Big(\lambda_t+\frac{Z^L_t}{L_{t-}}\Big).
\end{equation}
See Nutz~\cite[Corollary~3.12]{Nutz.09b} for these results. Finally, still under the assumption of continuity, the solution to the dual problem~\eqref{eq:dualProblem} is given by the local martingale
\begin{equation}\label{eq:DualOptimizerFormula}
  \hY=y_0\cE\Big(-\lambda\sint M + \frac{1}{L_{-}}\sint N^L\Big),
\end{equation}
with the constant $y_0=u'_p(x_0)=L_0 x_0^{p-1}$ (cf.~\cite[Remark~5.18]{Nutz.09b}).

\begin{Remark}\label{rk:Lcontinuity}
  Continuity of $S$ does not imply that $L$ is continuous; the local martingale $N^L$ may still have jumps
  (see also~\cite[Remark~3.13(i)]{Nutz.09b}).
 If the filtration $\F$ is continuous (i.e., all $\F$-martingales are continuous), it clearly follows that $L$ and $S$ are continuous. %
 The most important example with this property is the Brownian filtration.
\end{Remark}

\subsection{The Strategy for the Proofs}
We can now summarize the basic scheme that is common for the proofs of the three theorems.

The first step is to prove the \emph{pointwise convergence} of the opportunity process $L$ or of the dual opportunity process $L^*$; the choice of the process depends on the theorem. The convergence of the optimal propensity to consume $\hkappa$ then follows in view of the feedback formula~\eqref{eq:consumptionFeedback}. The definitions of $L$ and $L^*$ via the value processes lend themselves to control-theoretic arguments and of course Jensen's inequality will be the basic tool to derive estimates. In view of the relation $L^*=L^\beta$ from~\eqref{eq:DualAndPrimalOppProc}, it is essentially equivalent whether one works with $L$ or $L^*$, as long as $p$ is fixed. However, the dual problem has the advantage of being defined over a set of supermartingales, which are easier to handle than consumption and wealth processes. This is particularly useful when passing to the limit.

The second step is the convergence of the trading strategy $\hpi$. Note that its formula~\eqref{eq:optStrategy} contains the integrand $Z^L$ from the KW decomposition of $L$ with respect to~$M$. Therefore, the convergence of $\hpi$ is related to the \emph{convergence of the martingale part $M^L$} (resp.~$M^{L^*}$). In general, the pointwise convergence of a semimartingale is not enough to conclude the convergence of its martingale part; this requires some control over the semimartingale decomposition. In our case, this control is given by the Bellman BSDE~\eqref{eq:BellmanBSDEforL}, which can be seen as a description for the dependence of the finite variation part $A^L$ on the martingale part $M^L$.
As we use the BSDE to show the convergence of $M^L$, we benefit from techniques from the theory of quadratic BSDEs.
However, we cannot apply standard results from that theory since our assumptions are not strong enough.

In general, our approach is to extract as much information as possible by basic control arguments and convex analysis \emph{before} tackling the BSDE, rather than to rely exclusively on (typically delicate) BSDE arguments. For instance, we use the BSDE only after establishing the pointwise convergence of its left hand side, i.e., the opportunity process. This essentially eliminates the need for an \emph{a priori} estimate or a comparison principle and constitutes a key reason for the generality of our results.
Our procedure shares basic features of the viscosity approach to Markovian control problems, where one also works directly with the value function before tackling the Hamilton-Jacobi-Bellman equation.

\section{Auxiliary Results}\label{se:auxResults}

We start by collecting inequalities for the dependence of the opportunity processes on $p$.
The precise formulations are motivated by the applications in the proofs of the previous theorems,
but the comparison results are also of independent interest.

\subsection{Comparison Results}

We assume the entire section that $u_{p_0}(x_0)<\infty$ for a given exponent $p_0$. For convenience, we restate the quantities
$\beta=1/(1-p)>0$ and $q=\frac{p}{p-1}$ defined in~\eqref{eq:betaAndq}.
It is useful to note that
$q\in (-\infty,0)$ for $p\in (0,1)$ and vice versa. When there is a second exponent $p_0$ under consideration, $\beta_0$ and $q_0$ have the obvious definition. We also recall from~\eqref{eq:BoundsR} the bounds $k_1$ and $k_2$ for $D$.

\begin{Prop}\label{pr:ComparisonDualL}
 Let $0<p< p_0<1$. For each $t\in[0,T]$,
  \begin{align}\label{eq:ComparisonL1}
     L_t^*(p)\; & \leq\;E\Big[\int_t^T D_s^{\beta}\,\mu^\circ(ds)\Big|\cF_t\Big]^{1-q/q_0} \,\Big(k_1^{\beta-\beta_0} L^*_t(p_0)\Big)^{q/q_0},\\
     L_t(p)\;& \leq\;\big(k_2 \mu^\circ[t,T]\big)^{1-p/p_0} L_t(p_0)^{p/p_0}.\label{eq:ComparisonL2}
  \end{align}
 If $p< p_0<0$, the converse inequalities hold, if in~\eqref{eq:ComparisonL1} $k_1$ is replaced by $k_2$.
 If $p<0<p_0<1$, the converse inequalities hold, if in~\eqref{eq:ComparisonL2} $k_2$ is replaced by $k_1$.
\end{Prop}

\begin{proof}
  We fix $t$ and begin with~\eqref{eq:ComparisonL1}.
  To unify the proofs, we first argue a Jensen's inequality:
  if $X=(X_s)_{s\in[t,T]}>0$ is optional and $\alpha\in (0,1)$, then
  \begin{equation}\label{eq:JensenInProofComparison}
       E\Big[\int_t^T \hspace{-.5em} D_s^{\beta} X_s^\alpha\,\mu^\circ(ds)\Big|\cF_t\Big]
       \leq E\Big[\int_t^T \hspace{-.5em} D_s^{\beta}\mu^\circ(ds)\Big|\cF_t\Big]^{1-\alpha}
         E\Big[\int_t^T \hspace{-.5em} D_s^{\beta}X_s \,\mu^\circ(ds)\Big|\cF_t\Big]^\alpha.
  \end{equation}
  To see this, introduce the probability space
  $\big([t,T]\times\Omega,\cB([t,T])\otimes\cF, \nu \big)$, where
  \[
    \nu(I\times G):=E\Big[\xi^{-1} \int_I 1_G D_s^{\beta}\,\mu^\circ(ds)\Big],\quad G\in\cF,\, I\in \cB([t,T]),
  \]
  with the normalizing factor $\xi:=E[\int_t^T D_s^{\beta}\mu^\circ(ds)|\cF_t]$. On this space, $X$ is a random variable and we have the conditional Jensen's inequality
  \[
    E^\nu\big[X^\alpha\big|[t,T]\times \cF_t\big] \leq E^\nu\big[X\big|[t,T]\times \cF_t\big]^{\alpha}
  \]
  for the $\sigma$-field $[t,T]\times \cF_t:=\{[t,T]\times A:\,A\in\cF_t\}$.
  But this inequality coincides with~\eqref{eq:JensenInProofComparison} if we identify
  $L^0([t,T]\times\Omega,[t,T]\times \cF_t)$ and $L^0(\Omega,\cF_t)$ by using that an element of the first space is necessarily constant in its time variable.

  Let $0<p\leq p_0<1$ and let $\hY:=\hY(p_0)$ be the solution of the dual problem for $p_0$. Using~\eqref{eq:DualOppProcAltDef} and
  then~\eqref{eq:JensenInProofComparison} with $\alpha:=q/q_0\in (0,1)$ and $X_s^\alpha:=\big((\hY_s/\hY_t)^{q_0}\big)^{\alpha}=(\hY_s/\hY_t)^{q}$,
  \begin{align*}
    L_t^*(p) %
       &\leq E\Big[\int_t^T D_s^{\beta} \big({\hY_s/\hY_t}\big)^{q} \,\mu^\circ(ds)\Big|\cF_t\Big]\\
       &\leq E\Big[\int_t^T D_s^{\beta}\mu^\circ(ds)\Big|\cF_t\Big]^{1-q/q_0}  E\Big[\int_t^T D_s^{\beta} (\hY_s/\hY_t)^{q_0}\,\mu^\circ(ds)\Big|\cF_t\Big]^{q/q_0}.
  \end{align*}
  Now $D_s^{\beta}\leq k_1^{\beta-\beta_0}D_s^{\beta_0}$ since $\beta-\beta_0<0$, which completes the proof of the first claim in view of~\eqref{eq:DualOppProcAltDef}.
  In the cases with $p<0$, the infimum in \eqref{eq:DualOppProcAltDef} is replaced by a supremum and
  $\alpha=q/q_0$ is either $>1$ or $<0$, reversing the direction of Jensen's inequality.

  We turn to~\eqref{eq:ComparisonL2}. Let $0<p\leq p_0<1$ and $\hX=\hX(p)$, $\hc=\hc(p)$. Using~\eqref{eq:OppProcIndep} and (the usual) Jensen's inequality twice,
  \begin{align*}
    L_t(p_0)\hX_t^{p_0}
    & \geq E\Big[\int_t^T D_s \hc_s^{p_0}\,\mu^\circ(ds)\Big|\cF_t\Big]  \\
    & \geq \mu^\circ[t,T]^{1-p_0/p}   E\Big[\int_t^T D^{p/p_0}_s \hc_s^{p}\,\mu^\circ(ds)\Big|\cF_t\Big]^{p_0/p}\\
    & \geq \big(k_2\mu^\circ[t,T]\big)^{1-p_0/p} \big(L_t(p)\hX_t^{p}\big)^{p_0/p}
  \end{align*}
  and the claim follows. The other cases are similar.
\end{proof}

A useful consequence is that $L(p)$ gains moments as $p$ moves away from the possibly critical exponent $p_0$.

\begin{Cor}\label{co:ComparisonAppl}
  \begin{enumerate}[topsep=3pt, partopsep=0pt, itemsep=1pt,parsep=2pt]
    \item  Let $0<p< p_0<1$. Then
     \begin{equation}\label{eq:ComparisonAppl1}
      L(p)\leq C L(p_0)
     \end{equation}
     with a constant $C$ independent of $p_0$ and $p$. In the case without intermediate consumption we can take $C=1$.

    \item Let $r\geq1$ and $0<p\leq p_0/r$. Then
     \[
       E\big[(L_\tau(p))^r\big]\leq C_r
     \]
     for all stopping times $\tau$, with a constant $C_r$ independent of $p_0,p,\tau$. In particular, $L(p)$ is of class (D) for all $p\in(0,p_0)$.
   \end{enumerate}
\end{Cor}

\begin{proof}
  (i)~~Denote $L=L(p_0)$. By Lemma~\ref{le:BoundsForL}, $L/k_1\geq 1$, hence
  $L^{p/p_0}=k_1^{p/p_0} (L/k_1)^{p/p_0}\leq k_1^{p/p_0} (L/k_1)$ as $p/p_0\in (0,1)$.
  Proposition~\ref{pr:ComparisonDualL} yields the result with
  $C=\big (\mu^\circ[0,T] k_2/k_1\big)^{1-p/p_0}$; note that $C\leq 1\vee (1+T)k_2/k_1$. In the absence of intermediate consumption we may assume $k_1=k_2=1$ by the subsequent Remark~\ref{rk:measureChangeTrick} and then $C=1$.

  (ii)~~Let $r\geq1$, $0<p\leq p_0/r$, and $L=L(p_0)$. Proposition~\ref{pr:ComparisonDualL} shows
  \[
    L_t(p)^r \leq \big(k_2 \mu^\circ[t,T]\big)^{r(1-p/p_0)} L_t^{r p/p_0}\leq \big((1\vee k_2)(1+T)\big)^{r} L_t^{r p/p_0}.
  \]
  Note $r p/p_0\in(0,1)$, thus $L^{r p/p_0}$ is a supermartingale by Lemma~\ref{le:BoundsForL} and
  $E[L_\tau^{r p/p_0}] \leq L_0^{r p/p_0}\leq 1\vee k_2$.
\end{proof}

\begin{Remark}\label{rk:measureChangeTrick}
  In the case without intermediate consumption we may assume $D\equiv 1$ in the proof of Corollary~\ref{co:ComparisonAppl}(i). Indeed, $D$ reduces to the random variable $D_T$ and can be absorbed into the measure $P$ as follows. Under the measure $\widetilde{P}$ with $P$-density process $\xi_t=E[D_T|\cF_t]/E[D_T]$,
  the opportunity process for the utility function $\widetilde{U}(x)=\tfrac{1}{p}x^p$ is $\widetilde{L}=L/\xi$ by \cite[Remark~3.2]{Nutz.09a}. If Corollary~\ref{co:ComparisonAppl}(i) is proved for $D\equiv 1$, we conclude
  $\widetilde{L}(p)\leq \widetilde{L}(p_0)$ and then the inequality for $L$ follows.
\end{Remark}

Inequality~\eqref{eq:ComparisonAppl1} is stated for reference as it has a simple form; however, note that it was deduced using the very poor estimate $a^b\geq a$ for $a,b\geq1$.
In the pure investment case, we have $C=1$ and so~\eqref{eq:ComparisonAppl1} is a direct comparison result.
Intermediate consumption destroys this monotonicity property:~\eqref{eq:ComparisonAppl1} fails for $C=1$ in that case, e.g., if
$D\equiv 1$ and $R_t=t+W_t$, where $W$ is a standard Brownian motion, and $p=0.1$ and $p_0=0.2$, as can be seen by explicit calculation.
This is not surprising from a BSDE perspective, because the driver of~\eqref{eq:BellmanBSDEforL} is not monotone
with respect to $p$ in the presence of the $d\mu$-term. In the pure investment case, the driver is monotone and so the comparison result can be expected, even for the entire parameter range. This is confirmed by the next result;
note that the inequality is \emph{converse} to~\eqref{eq:ComparisonL2} for the considered parameters.

\begin{Prop}\label{pr:ComparisonLNoCons}
 Let $p< p_0<0$, then
  \[
    L_t(p)\leq \frac{k_2}{k_1} \, \big(\mu^\circ[t,T]\big)^{p_0-p} L_t(p_0).
  \]
  In the case without intermediate consumption, $L(p)\leq L(p_0)$.
\end{Prop}

The proof is based on the following auxiliary statement.
\begin{Lemma}\label{le:monotoneFunctionRevHolder}
  Let $Y>0$ be a supermartingale. For fixed $0\leq t\leq s\leq T$, %
  \[
    \phi: \, (0,1)\to \R_+,\quad q\mapsto \phi(q):=\Big(E\big[(Y_s/Y_t)^q\big|\cF_t\big]\Big)^{\frac{1}{1-q}}
  \]
  is a monotone decreasing function $P$-a.s. If $Y$ is a martingale, we have
  $\phi(1):=\lim_{q\to 1-}\phi(q)=\exp\big(-E\big[(Y_s/Y_t)\log(Y_s/Y_t)\big|\cF_t\big]\big)$ $P$-a.s.,
  where the conditional expectation has values in $\R\cup\{+\infty\}$.
\end{Lemma}
Lemma~\ref{le:monotoneFunctionRevHolder} can be obtained using Jensen's inequality and a suitable change of measure; we refer to~\cite[Lemma~4.10]{Nutz.09a} for details.

\begin{proof}[Proof of Proposition~\ref{pr:ComparisonLNoCons}.]
  Let $0<q_0<q<1$ be the dual exponents and
  denote $\hY:=\hY(p)$. By Lemma~\ref{le:monotoneFunctionRevHolder} and Jensen's inequality for $\frac{1-q}{1-q_0}\in (0,1)$,
  \begin{align*}
    \int_t^T E\big[ & (\hY_s/\hY_t)^{q} \big|\cF_t\big]\,\mu^\circ(ds)
     \leq \int_t^T \Big(E\big[(\hY_s/\hY_t)^{q_0}\big|\cF_t\big]\Big)^{\frac{1-q}{1-q_0}}\,\mu^\circ(ds)\\
    & \leq \mu^\circ[t,T]^{\big(1-\frac{1-q}{1-q_0}\big)} \bigg(\int_t^T E\big[(\hY_s/\hY_t)^{q_0}\big|\cF_t\big]\,\mu^\circ(ds)\bigg)^{\frac{1-q}{1-q_0}}.%
  \end{align*}
  Using~\eqref{eq:BoundsR} and~\eqref{eq:DualOppProcAltDef} twice, we conclude that
  \begin{align*}
    L^*_t(p)&\leq k_2^\beta \int_t^T E\big[(\hY_s/\hY_t)^{q} \big|\cF_t\big]\,\mu^\circ(ds)\\
    & \leq k_2^\beta k_1^{-\beta_0\frac{1-q}{1-q_0}} \mu^\circ[t,T]^{\big(1-\frac{1-q}{1-q_0}\big)} \bigg(\int_t^T E\big[D_s^{\beta_0}(\hY_s/\hY_t)^{q_0}\big|\cF_t\big]\,\mu^\circ(ds)\bigg)^{\frac{1-q}{1-q_0}}\\
    & \leq  k_2^\beta k_1^{-\beta_0\frac{1-q}{1-q_0}} \mu^\circ[t,T]^{\big(1-\frac{1-q}{1-q_0}\big)} L^*_t(p_0)^{\frac{1-q}{1-q_0}}.
  \end{align*}
  Now~\eqref{eq:DualAndPrimalOppProc} and $\beta=1-q$ yield the first result.
  In the case without intermediate consumption,
  we may assume $D\equiv1$ and hence $k_1=k_2=1$, as in Remark~\ref{rk:measureChangeTrick}.
\end{proof}

\begin{Remark}
   Our argument for Proposition~\ref{pr:ComparisonLNoCons} extends to $p=-\infty$ (cf.\ Lemma~\ref{le:LpGreaterLexp} below).  The proposition generalizes
   \cite[Proposition~2.2]{ManiaTevzadze.03}, where the result is proved for the case without intermediate consumption
   and under the additional condition that $\hY(p_0)$ is a martingale
   (or equivalently, that the $q_0$-optimal equivalent martingale measure exists).
\end{Remark}

Propositions~\ref{pr:ComparisonDualL} and~\ref{pr:ComparisonLNoCons} combine to the following continuity property of $p\mapsto L(p)$ at interior points of $(-\infty,0)$. We will not pursue this further as we are interested mainly in the boundary points of this interval.

\begin{Cor}
  Assume $D\equiv 1$ and let $C_t:=\mu^\circ[t,T]$. If $p\leq p_0<0$,
  \[
    C_t^{1-p/p_0} L(p_0)^{p/p_0}\,\leq\, L(p)\,\leq\, C_t^{p_0-p} L(p_0) \,\leq\, C_t^{1-p_0/p+p_0-p} L(p)^{p_0/p}.
  \]
  In particular, $p\mapsto L_t(p)$ is continuous on $(-\infty,0)$ uniformly in $t$, $P$-a.s.
\end{Cor}

\begin{Remark}
  The optimal propensity to consume $\hkappa(p)$ is \emph{not} monotone with respect to $p$ in general. For instance, monotonicity fails
  for $D\equiv 1$ and $R_t=t+W_t$, where $W$ is a standard Brownian motion, and $p\in \{-1/2,-1,-2\}$.
  One can note that $p$ determines both the risk aversion and the elasticity of intertemporal substitution (see, e.g., Gollier~\cite[\S15]{Gollier.01}).
  As with any time-additive utility specification, it is not possible in our setting to study the dependence
  on each of these quantities in an isolated way.
\end{Remark}

\subsection{$BMO$ Estimate}
In this section we give $BMO$ estimates for the martingale part of $L$.
The following lemma is well known; we state the proof since the argument will be used also later on.

\begin{Lemma}\label{le:QuadVarLemma}
  Let $X$ be a submartingale satisfying $0\leq X \leq \alpha$ for some constant $\alpha>0$.
  Then for all stopping times $0\leq \sigma\leq \tau\leq T$,
  \[
    E\big[[X]_\tau-[X]_\sigma \big|\cF_{\sigma}\big] \leq E\big[ X^2_\tau-X^2_\sigma \big|\cF_{\sigma}\big].
  \]
\end{Lemma}

\begin{proof}
  Let $X=X_0+M^X+A^X$ be the Doob-Meyer decomposition.
  As $X_t^2=X_0^2 +2\int_0^t X_{s-}\,(dM^X_s+dA^X_s)+[X]_t$ and $2\int_\sigma^\tau X_{s-}dA^X_s\geq0$,
  \[
    [X]_\tau-[X]_\sigma%
    \leq X_\tau^2-X_\sigma^2 - 2\int_\sigma^\tau X_{s-}dM^X_s.
  \]
  The claim follows by taking conditional expectations because $X_{-}\sint M^X$ is a martingale. Indeed, $X$ is bounded and $\sup_t |M^X_t|\leq 2\alpha+ A^X_T \in L^1$,
  so the BDG inequalities~\cite[VII.92]{DellacherieMeyer.82} show $[M^X]_T^{1/2}\in L^1$, hence $[X_-\sint M^X]_T^{1/2}\in L^1$, which by the BDG inequalities implies that $\sup_t|X_-\sint M^X_t|\in L^1$.
\end{proof}

We wish to apply Lemma~\ref{le:QuadVarLemma} to $L(p)$ in the case $p<0$. However, the submartingale property fails in general for the case with intermediate consumption (cf.~Lemma~\ref{le:BoundsForL}). We introduce instead a closely related process having this property.

\begin{Lemma}\label{le:TheSubmartB}
  Let $p<0$ and consider the case with intermediate consumption. Then
  \[
    B_t:= \Big(\frac{1+T-t}{1+T}\Big)^p L_t+ \frac{1}{(1+T)^p} \int_0^t D_s\,ds
  \]
  is a submartingale satisfying
  $
    0<B_t\leq k_2 (1+T)^{1-p}.
  $
\end{Lemma}

\begin{proof}
  Choose $(\pi,c)\equiv (0,x_0/(1+T))$ in~\cite[Proposition~3.4]{Nutz.09a} to see that $B$ is a submartingale.
  The bound follows from Lemma~\ref{le:BoundsForL}.
\end{proof}

We are now in the position to exploit Lemma~\ref{le:QuadVarLemma}.

\begin{Lemma}\label{le:LisBMO}
  \begin{enumerate}[topsep=3pt, partopsep=0pt, itemsep=1pt,parsep=2pt]
   \item Let $p_1<0$. There exists a constant $C=C(p_1)$ such that $\|M^{L(p)}\|_{BMO}\leq C$ for all $p\in (p_1,0)$.
   In the case without intermediate consumption one can take $p_1=-\infty$.

   \item Assume $u_{p_0}(x_0)<\infty$ for some $p_0\in(0,1)$ and let $\sigma$ be a stopping
    time such that $L(p_0)^\sigma \leq \alpha$ for a constant $\alpha>0$.
    Then there exists $C'=C'(\alpha)$ such that $\|(M^{L(p)})^\sigma\|_{BMO}\leq C'$ for all $p\in (0,p_0]$.
  \end{enumerate}
\end{Lemma}

\begin{proof}
  (i)~~Let $p_1<p<0$ and let $\tau$ be a stopping time. We first show that
  \begin{equation}\label{eq:proofLisBMO}
    E\big[[L(p)]_T-[L(p)]_\tau\big|\cF_\tau\big]\leq C.
  \end{equation}
  In the case without intermediate consumption, $L=L(p)$ is a positive submartingale with $L\leq k_2$ (Lemma~\ref{le:BoundsForL}), so Lemma~\ref{le:QuadVarLemma} implies~\eqref{eq:proofLisBMO} with $C= k_2^2$.
  In the other case, define $B$ as in Lemma~\ref{le:TheSubmartB} and
  $f(t):=(\frac{1+T-t}{1+T})^p$. Then $[L]_t-[L]_0=\int_0^t f^{-2}(s)\,d[B]_s$ and $f^{-2}(s)\leq 1$ as $f$ is increasing with $f(0)=1$. Thus
  $
    [L]_T-[L]_\tau=\int_\tau^T f^{-2}(s)\,d[B]_s\leq [B]_T-[B]_\tau
  $.
  Now~\eqref{eq:proofLisBMO} follows since $B\leq k_2(1+T)^{1-p}$ and Lemma~\ref{le:QuadVarLemma} imply
  \[
    E\big[[B]_T-[B]_\tau\big|\cF_\tau\big]\leq k_2^2 (1+T)^{2-2p}\leq k_2^2 (1+T)^{2-2p_1} =: C(p_1).
  \]
  We have $[L]=L_0^2+[M^L]+[A^L]+2[M^L,A^L]$. Since $A^L$ is predictable, $N:=2[M^L,A^L]$ is a local martingale with some localizing sequence $(\sigma_n)$.
  Moreover, $[M^L]_t-[M^L]_s=[L]_t-[L]_s - ([A]_t-[A]_s) - (N_t-N_s)$ and~\eqref{eq:proofLisBMO} imply
  \[
    E\big[[M^L]_{T\wedge \sigma_n}-[M^L]_{\tau\wedge \sigma_n}\big|\cF_{\tau\wedge \sigma_n}\big]\leq C.
  \]
  Choosing $\tau=0$ and $n\to\infty$ we see that $[M]_T\in L^1(P)$ and thus Hunt's Lemma~\cite[V.45]{DellacherieMeyer.82} shows the a.s.-convergence
  in this inequality; i.e., we have $E\big[[M^L]_T-[M^L]_{\tau}\big|\cF_{\tau}\big]\leq C$.
  If $L$ is bounded by $\alpha$, the jumps of $M^L$ are bounded by $2\alpha$ (cf.~\cite[I.4.24]{JacodShiryaev.03}), therefore
  \[
    \sup_\tau E\big[[M^L]_T-[M^L]_{\tau-}\big|\cF_{\tau}\big]\leq C+4\alpha^2.
  \]
  By Lemma~\ref{le:BoundsForL} we can take $\alpha=k_2(1+T)^{1-p_1}$, and $\alpha=k_2$ when there is no intermediate consumption.

  (ii)~~Let $0<p\leq p_0<1$. The assumption and Corollary~\ref{co:ComparisonAppl}(i) show that $L(p)^\sigma\leq C_\alpha$ for a constant $C_\alpha$ independent of $p$ and $p_0$.
  We apply Lemma~\ref{le:QuadVarLemma} to the nonnegative process $X(p):=C_\alpha-L(p)^\sigma$, which is a submartingale by Lemma~\ref{le:BoundsForL}, and obtain
  $
    E\big[[L(p)^\sigma]_T-[L(p)^\sigma]_\tau\big|\cF_\tau\big]=E\big[[X(p)]_T-[X(p)]_\tau\big|\cF_\tau\big] \leq C_\alpha^2.
  $
  Now the rest of the proof is as in (i).
\end{proof}

\begin{Cor}\label{co:lambdaBMO}
  Let $S$ be continuous and assume that either $p\in (0,1)$ and $L$ is bounded or that $p<0$ and $L$ is bounded away from zero.
  Then $\lambda\sint M\in BMO$, where $\lambda$ and $M$ are defined by~\eqref{eq:StructureContForR}.
\end{Cor}
\begin{proof}
  In both cases, the assumed bound and
  Lemma~\ref{le:BoundsForL} imply that $L$ is bounded away from zero and infinity. Taking conditional expectations in~\eqref{eq:BellmanBSDEforL}, we obtain a constant $C>0$ such that
  \[
    E\bigg[\int_t^T L_{-}\Big(\lambda+\frac{Z^L}{L_{-}}\Big)^\top\,d\br{M}\, \Big(\lambda+\frac{Z^L}{L_{-}}\Big)\bigg|\cF_t\bigg]\leq C,\quad 0\leq t\leq T.
  \]
  Moreover, we have $M^L\in BMO$ by Lemma~\ref{le:LisBMO}. Using the bounds for $L$ and the Cauchy-Schwarz inequality, it follows that
  $E[\int_t^T \lambda^\top\,d\br{M}\,\lambda |\cF_t]\leq C'(1+\|Z^L\sint M\|_{BMO})\leq C'(1+\|M^L\|_{BMO})$ for a constant $C'>0$.
\end{proof}

We remark that uniform bounds for $L$ (as in the condition of Corollary~\ref{co:lambdaBMO}) are equivalent to a reverse H\"older inequality $\mathrm{R}_q(P)$ for some element of the dual domain $\sY$; see~\cite[Proposition~4.5]{Nutz.09a} for details. Here the index $q$ satisfies $q<1$. Therefore, our corollary complements well known results stating that $\mathrm{R}_q(P)$ with $q>1$ implies $\lambda\sint M\in BMO$ (in a suitable setting); see, e.g., Delbaen et al.~\cite[Theorems A,B]{DelbaenEtAl.97}. %

\section{The Limit $p\to-\infty$}\label{se:pto-inftyProofs}
The first goal of this section is to prove Theorem~\ref{th:Limit-inftyEconomics}. Recall that the consumption strategy is related to the opportunity processes via~\eqref{eq:consumptionFeedback} and~\eqref{eq:DualAndPrimalOppProc}. From these relations and the intuition mentioned before Theorem~\ref{th:Limit-inftyEconomics}, we expect that the dual opportunity process $L_t^*=L^\beta_t$ converges to $\mu^\circ[t,T]$ as $p\to -\infty$.
Noting that the exponent $\beta=1/(1-p)\to 0$, this implies that $L_t(p)\to\infty$ for all $t<T$, in the case with intermediate consumption. Therefore, we shall work here with  $L^*$ rather than $L$.
In the pure investment case, the situation is different as then $L\leq k_2$ (Lemma~\ref{le:BoundsForL}). There, the limit of $L$ yields additional information; this is examined in Section~\ref{se:ExponentialLimit} below.

\begin{Prop}\label{pr:DualLconv-infty}
  For each $t\in[0,T]$,
  \[
    \lim_{p\to-\infty} L_t^*(p) = \mu^\circ[t,T]\quad P\mbox{-a.s.\ and in }L^r(P),\;r\in [1,\infty),
  \]
  with a uniform bound. %
  If $\F$ is continuous, the convergences are uniform in $t$.
\end{Prop}

\begin{Remark}\label{rk:convAtStoppingtime}
  We will use later that the same convergences hold if $t$ is replaced by a stopping time, which is an immediate consequence
  in view of the uniform bound. Of course, we mean by ``uniform bound'' that there exists a constant $C>0$, independent of $p$ and $t$, such that $0\leq L_t^*(p)\leq C$. Analogous terminology will be used in the sequel.
\end{Remark}

\begin{proof}
  We consider $0>p\to -\infty$ and note that $q\to 1-$ and $\beta\to 0+$.
  From Lemma~\ref{le:BoundsForL} we have
  \begin{equation}\label{eq:ProofDualLconv-infty}
    0\leq L^*_t(p)=L^\beta_t(p) \leq k_2^{\beta} \mu^\circ[t,T] \to \mu^\circ[t,T],
  \end{equation}
  uniformly in $t$. To obtain a lower bound, we consider the density process $Y$ of some $Q\in\sM$, which exists by assumption~\eqref{eq:ELMMexists}. From~\eqref{eq:DualOppProcAltDef} we have
  \begin{align*}
    L^*_t(p) %
     \geq k_1^{\beta} \int_t^T E\big[(Y_s/Y_t)^{q}\big|\cF_t\big]\,\mu^\circ(ds).
  \end{align*}
  For fixed $s\geq t$, clearly $(Y_s/Y_t)^{q}\to Y_s/Y_t$ $P$-a.s.\ as $q\to1$, and noting the bound
  $0\leq (Y_s/Y_t)^{q}\leq 1+Y_s/Y_t\in L^1(P)$ we conclude by dominated convergence that
  \[
    E\big[(Y_s/Y_t)^{q}\big|\cF_t\big]\to E\big[Y_s/Y_t\big|\cF_t\big]\equiv1\quad P\mbox{-a.s.,\;\;for all } s\geq t.
  \]
  Since $Y^{q}$ is a supermartingale, $0\leq E\big[(Y_s/Y_t)^{q}\big|\cF_t\big] \leq 1$. Hence, for each $t$,
  dominated convergence shows
  \[
    \int_t^T E\big[(Y_s/Y_t)^{q}\big|\cF_t\big]\,\mu^\circ(ds)\to \mu^\circ[t,T] \;\;P\mbox{-a.s.}
  \]
  This ends the proof of the first claim. The convergence in $L^r(P)$ follows by the bound~\eqref{eq:ProofDualLconv-infty}.

  Assume that $\F$ is continuous; then the martingale $Y$ is continuous. For fixed $(s,\omega)\in[0,T]\times\Omega$ we consider (a version of)
  \[
    f_q(t):=E\big[(Y_s/Y_t)^{q}\big|\cF_t\big]^{1/q}(\omega), \quad t\in[0,s].
  \]
  These functions are continuous in $t$ and increasing in $q$ by Jensen's inequality, and converge to $1$ for each $t$.  Hence $f_q\to 1$ uniformly in $t$ on the compact $[0,s]$, by Dini's lemma. The same holds for $f_q(t)^q=E\big[(Y_s/Y_t)^{q}\big|\cF_t\big](\omega)$.

  Fix $\omega\in\Omega$ and let $\eps,\eps'>0$. By Egorov's theorem there exist a measurable set $I=I(\omega)\subseteq [0,T]$ and $\delta=\delta(\omega)\in(0,1)$ such that $\mu^\circ([0,T]\setminus I)<\eps$ and
  $\sup_{t\in [0,s]}|E\big[(Y_s/Y_t)^{q}\big|\cF_t\big]-1|<\eps'$ for all $q>1-\delta$
  and all $s\in I$. For $q>1-\delta$ and $t\in [0,T]$ we have
  \begin{align*}
    & \int_t^T \big|E\big[(Y_s/Y_t)^{q}\big|\cF_t\big]-1\big| \,\mu^\circ(ds)\\
    &\leq   \int_{I} \big|E\big[(Y_s/Y_t)^{q}\big|\cF_t\big]-1\big| \,\mu^\circ(ds)
    + \int_{[t,T]\setminus I} \big|E\big[(Y_s/Y_t)^{q}\big|\cF_t\big]-1\big| \,\mu^\circ(ds)\\
    &\leq \eps'(1+T) + \eps.
  \end{align*}
  We have shown that $\sup_{t\in[0,T]}|L^*_t(p)-\mu^\circ[t,T]|\to 0$ $P$-a.s.,
  and also in $L^r(P)$ by dominated convergence and the uniform bound resulting from~\eqref{eq:ProofDualLconv-infty} in view of
  $k_2^{\beta} \mu^\circ[t,T]\leq (1\vee k_2)(1+T)$.
\end{proof}

Under additional continuity assumptions, we will prove that the martingale part of $L^*$ converges to zero in $\cH^2_{loc}$.
We first need some preparations.
For each $p$, it follows from Lemma~\ref{le:BoundsForL} that $L^*$ has a canonical decomposition
$L^*=L^*_0+ M^{L^*} + A^{L^*}$. When $S$ is continuous, we denote the KW decomposition with respect to $M$ by $L^*=L^*_0+ Z^{L^*}\sint M + N^{L^*} + A^{L^*}$. If in addition $L$ is continuous, we obtain from $L^*=L^\beta$ and~\eqref{eq:FVpartofL} by It\^o's formula that
\begin{align}\label{eq:dualOppFVformulas}
  M^{L^*} & = \beta L^{\beta-1}\sint M^L;\quad\;   Z^{L^*}/L^*  = \beta Z^L/L; \quad\;  N^{L^*}=\beta L^{\beta-1}\sint N^L;  \\
  A^{L^*}&= \tfrac{q}{2} \int \big(\beta\lambda L^* \,+2Z^{L^*}\big)^\top\,d\br{M}\,\lambda +\tfrac{p}{2} \int \big(L^*\big)^{-1}\,d\br{N^{L^*}} -\int D^\beta\,d\mu.\nonumber
\end{align}
Here $d\mu$ is a shorthand for $\mu(ds)$.

\begin{Lemma}\label{le:uniformStopping-infty}
  Let $p_0<0$. There exists a localizing sequence $(\sigma_n)$ such that
  \[
    \big(L^*(p)\big)^{\sigma_n}_- > 1/n\;\;\mbox{simultaneously for all }p\in (-\infty,p_0];
  \]
  and moreover, if $S$ and $L(p)$ are continuous, $(M^{L^*(p)})^{\sigma_n}\in BMO$ for $p\leq p_0$.
\end{Lemma}

\begin{proof}
  Fix $p_0<0$ (and corresponding $q_0$) and a sequence $\eps_n\downarrow 0$ in $(0,1)$. Set
  $\sigma_n=\inf \{t\geq 0:\,L^*_t(p_0)\leq \eps_n\}\wedge T$.
  Then $\sigma_n\to T$ stationarily because each path of $L^*(p_0)$ is bounded away from zero (Lemma~\ref{le:BoundsForL}).
  Proposition~\ref{pr:ComparisonDualL} implies that there is a constant $\alpha=\alpha(p_0)>0$ such that
  $L^*_t(p)\geq \alpha\, \big(L^*_t(p_0)\big)^{q/q_0}$ for all $p\leq p_0$. It follows that
  $L^*_{(\sigma_n\wedge t)-}(p)\geq \alpha  \eps_n^{1/q_0}$ for all $p\leq p_0$ and
  we have proved the first claim.

  Fix $p\in (-\infty,p_0]$, let $S$ and $L=L(p)$ be continuous and recall that
  $M^{L^*}=\beta L^{\beta-1}\sint M^L$ from~\eqref{eq:dualOppFVformulas}. Noting that $\beta-1<0$, we have just shown that the integrand $\beta L^{\beta-1}$ is bounded on $[0,\sigma_n]$. Since $M^L\in BMO$ by Lemma~\ref{le:LisBMO}(i), we conclude that $(M^{L^*})^{\sigma_n}\in BMO$. %
\end{proof}

\begin{Prop}\label{pr:StrategyConv}
  Assume that $S$ and $L(p)$ are continuous for all $p<0$. As $p\to-\infty$,
  \[
    Z^{L^*(p)}\to 0 \mbox{ in }L^2_{loc}(M)\quad\mbox{and} \quad N^{L^*(p)}\to0 \mbox{ in }\cH^2_{loc}.
  \]
\end{Prop}

\begin{proof}%
  We fix some $p_0<0$ and consider $p\in(-\infty,p_0]$. Using Lemma~\ref{le:uniformStopping-infty}, we may assume by localization
  that $M^{L^*(p)}\in \cH^2$ and $\lambda \in L^2(M)$.
  Define the continuous processes $X=X(p)$ by
  \[
    X_t(p):=k_2^\beta \mu^\circ[t,T]-L^*_t(p),
  \]
  then $0\leq X(p)\leq (1\vee k_2)(1+T)$ by~\eqref{eq:ProofDualLconv-infty}.
  Fix $p$. We shall apply It\^o's formula to $\Phi(X)$, where
  \[
    \Phi(x):=\exp(x)-x.
  \]
  For $x\geq0$, $\Phi$ satisfies
  \[
    \Phi(x)\geq 1, \q \Phi'(0)=0,\q \Phi'(x)\geq 0,\q\Phi''(x)\geq1,\q \Phi''(x)-\Phi'(x)=1.
  \]
  We have
  $
    \tfrac{1}{2}\int_0^T \Phi''(X)\,d\br{X}=\Phi(X_T)-\Phi(X_0)-\int_0^T \Phi'(X)\,(\,dM^X+dA^X).
  $
  As $\Phi'(X)$ is like $X$ uniformly bounded and $M^X=-M^{L^*}\in\cH^2$, the stochastic integral wrt.~$M^X$ is a true martingale and
  \[
    E\Big[\int_0^T \Phi''(X)\,d\br{X}\Big]=2E\big[\Phi(X_T)-\Phi(X_0)\big]-2E\Big[\int_0^T \Phi'(X)\,dA^X\Big].
  \]
  Note that $dA^X=-k_2^\beta\,d\mu-dA^{L^*}$, so that~\eqref{eq:dualOppFVformulas} yields
  \[
    2\,dA^X=-q\big(\beta\lambda L^* +2Z^{L^*}\big)^\top\,d\br{M}\,\lambda
      -p \big(L^*\big)^{-1}\,d\br{N^{L^*}} +2 \big(D^\beta-k_2^\beta\big) \,d\mu.
  \]
  Letting $p\to-\infty$, we have $q\to1-$ and $\beta\to 0+$. Hence, using that $X$ and $L^*$ are bounded uniformly in $p$,
  \begin{align*}
    -q\beta &E\Big[\int_0^T \Phi'(X) (\lambda L^*)^\top\,d\br{M}\,\lambda\Big]\to 0,\\
       &E\Big[\int_0^T \Phi'(X)\big(D^\beta-k_2^\beta\big) \,d\mu \Big]\to 0,\\
       &E\big[\Phi(X_T)-\Phi(X_0)\big]\to 0,
  \end{align*}
  where the last convergence is due to Proposition~\ref{pr:DualLconv-infty} (and the subsequent remark).
  If $o$ denotes the sum of these three expectations tending to zero,
    \begin{align*}
     E\bigg[\int_0^T &\Phi''(X)\,d\br{X}\bigg] \\
    & =E\bigg[\int_0^T \Phi'(X)\Big\{ 2q \big(Z^{L^*}\big)^\top\,d\br{M}\,\lambda + p \big(L^*\big)^{-1}\,d\br{N^{L^*}} \Big\} \bigg]+ o.
  \end{align*}
  Note $d\br{X}=d\br{L^*}=\big(Z^{L^*}\big)^\top\,d\br{M}\,Z^{L^*}+ d\br{N^{L^*}}$.
  For the right hand side, we use $\Phi'(X)\geq0$ and $|q|<1$ and the Cauchy-Schwarz inequality
  to obtain
  \begin{align*}
    & E\bigg[\int_0^T \Phi''(X)\,\Big\{ \big(Z^{L^*}\big)^\top\,d\br{M}\,Z^{L^*}+ d\br{N^{L^*}} \Big\}\bigg] \\
    & \leq E\bigg[\int_0^T \hspace{-.3em}\Phi'(X)\Big\{\big(Z^{L^*}\big)^\top\,d\br{M}\,Z^{L^*}\hspace{-.3em}+ \lambda^\top\,d\br{M}\,\lambda
    + p \big(L^*\big)^{-1}\,d\br{N^{L^*}} \Big\} \bigg]+ o.
  \end{align*}
  We bring the terms with $Z^{L^*}$ and $N^{L^*}$ to the left hand side, then
  \begin{align*}
   & E\bigg[\int_0^T \big\{\Phi''(X)-\Phi'(X) \big\} \big(Z^{L^*}\big)^\top\,d\br{M}\,Z^{L^*}\bigg]\\
   & + E\bigg[\int_0^T \hspace{-.5em}\big\{\Phi''(X)-p\Phi'(X) \big(L^*\big)^{-1}\big\}d\br{N^{L^*}} \bigg]
    \leq E\bigg[\int_0^T \hspace{-.4em}\Phi'(X)\lambda^\top d\br{M}\,\lambda \bigg]+ o.
 \end{align*}
 As $\Phi'(0)=0$, we have $\lim_{p\to-\infty} \Phi'(X_t)\to0$ $P$-a.s.~for all $t$, with a uniform bound, hence $\lambda \in L^2(M)$ implies that the right hand side converges to zero.
 We recall $\Phi''-\Phi'\equiv1$ and $\Phi''(X)-p\Phi'(X) \big(L^*\big)^{-1}\geq \Phi''(0)=1$. Whence
 both expectations on the left hand side are nonnegative and we can conclude that they converge to zero; therefore,
 $E\big[\int_0^T (Z^{L^*})^\top\,d\br{M}\,Z^{L^*}\big]\to 0$ and
  $E[\br{N^{L^*}}_T]\to 0$.
\end{proof}

\begin{proof}[Proof of Theorem~\ref{th:Limit-inftyEconomics}]
  In view of~\eqref{eq:consumptionFeedback}, part~(i) follows from Proposition~\ref{pr:DualLconv-infty};
  note that the convergence in $\cR^r_{loc}$ is immediate as $\hkappa(p)$ is locally bounded
  uniformly in $p$ by Lemma~\ref{le:uniformStopping-infty} and~\eqref{eq:consumptionFeedback}.
  For part~(ii), recall from~\eqref{eq:optStrategy} and~\eqref{eq:dualOppFVformulas} that
  \[
    \hpi=\beta (\lambda+ Z^L/L )
    =\beta\lambda +  Z^{L^*}/L^*
  \]
  for each $p$.
  As $\beta\to0$, clearly $\beta\lambda\to0$ in $L^2_{loc}(M)$.
  By Lemma~\ref{le:uniformStopping-infty}, $1/L^*$ is locally bounded
  uniformly in $p$, hence $\hpi(p)\to0$ in $L^2_{loc}(M)$ follows from Proposition~\ref{pr:StrategyConv}.
  As $\hkappa(p)$ is locally bounded uniformly in $p$, Corollary~\ref{co:ForWealthProcConv}(i)
  from the Appendix yields the convergence of the wealth processes $\hX(p)$.
\end{proof}

\subsection{Convergence to the Exponential Problem}\label{se:ExponentialLimit}

In this section, we prove Theorem~\ref{th:strategyConvExponentialMainRes}
and establish the convergence of the corresponding opportunity processes.
We assume that there is no intermediate consumption, that $S$ is locally bounded and satisfies~\eqref{eq:finiteEntropy}, and that the contingent claim $B$ is bounded (we will choose a specific $B$ later).
Hence there exists an (essentially unique) optimal strategy $\hvartheta\in\Theta$ for~\eqref{eq:exponentialProblem}. It is easy to see that $\hvartheta$ does not depend on the initial capital $x_0$. If $\vartheta\in\Theta$,
we denote by $G(\vartheta)=\vartheta \sint R$ the corresponding gains process and define
$\Theta(\vartheta,t)=\{\tvartheta\in\Theta:\, G_t(\tvartheta)=G_t(\vartheta)\}$.
We consider the value process (from initial wealth zero) of~\eqref{eq:exponentialProblem},
\[
  V_t(\vartheta):=\esssup_{\tvartheta\in\Theta(\vartheta,t)} E\big[-\exp\big(B - G_T(\tvartheta)\big)\big|\cF_t\big], \quad 0\leq t\leq T.
\]
Note the concatenation property $\vartheta^1, \vartheta^2\in \Theta \Rightarrow \vartheta^1 1_{[0,t]} + \vartheta^2 1_{(t,T]} \in \Theta$. With $G_{t,T}(\vartheta):=\int_t^T \vartheta \,dR$, we have
$G_T(\tvartheta)=G_t(\vartheta) + G_{t,T}(\tvartheta 1_{(t,T]})$ for $\tvartheta\in\Theta(\vartheta,t)$.
Therefore, if we define the exponential opportunity process
\begin{equation}\label{eq:LexpDef}
    L^{\exp}_t:=\essinf_{\tvartheta\in\Theta} E\big[\exp\big(B - G_{t,T}(\tvartheta)\big)\big|\cF_t\big], \quad 0\leq t\leq T,
\end{equation}
then using standard properties of the essential infimum one can check that
\[
 V_t(\vartheta)=-\exp(-G_t(\vartheta))\,L^{\exp}_t.
\]
Thus $L^{\exp}$ is a reduced form of the value process, analogous to $L(p)$ for power utility. We also note that $L^{\exp}_T=\exp(B)$.

\begin{Lemma}\label{le:BoundsForLexp}
  The exponential opportunity process $L^{\exp}$ is a submartingale satisfying $L^{\exp}\leq \|\exp(B)\|_{L^\infty(P)}$ and $L^{\exp}, L_-^{\exp}>0$.
\end{Lemma}

\begin{proof}
  The martingale optimality principle of dynamic programming is proved here exactly as, e.g., in~\cite[Proposition~A.2]{Nutz.09a}, and  yields that $V(\vartheta)$ is a supermartingale for every $\vartheta\in\Theta$ such that $E[V_{\cdot}(\vartheta)]>-\infty$ and a martingale if and only if $\vartheta$ is optimal.
  As $V(\vartheta)=-\exp(-G(\vartheta))\,L^{\exp}$, we obtain the submartingale property by the choice $\vartheta\equiv 0$.
  It follows that $L^{\exp}\leq \|L_T^{\exp}\|_{L^\infty}=\|\exp(B)\|_{L^\infty}$.

  The optimal strategy $\hvartheta$ is optimal for all the conditional problems~\eqref{eq:LexpDef}, hence
  $L_t^{\exp}=E\big[\exp\big(B - G_{t,T}(\hvartheta)\big)\big|\cF_t\big]>0$.
  Thus $\xi:=\exp(-G(\hvartheta))\,L^{\exp}$ is a positive martingale, by the optimality principle. In particular, we have $P[\inf_{0\leq t\leq T} \xi_t >0]=1$, and
  now the same property for $L^{\exp}$ follows.
\end{proof}

Assume that $S$ is continuous and denote the KW decomposition of $L^{\exp}$ with respect to $M$ by
$L^{\exp} = L^{\exp}_0 + Z^{L^{\exp}}\sint M + N^{L^{\exp}} + A^{L^{\exp}}$.
Then the triplet $(\ell,z,n):=\big(L^{\exp}, Z^{L^{\exp}},N^{L^{\exp}}\big)$ satisfies the BSDE
\begin{align}\label{eq:BellmanBsdeExp}
  d\ell_t & = \frac{1}{2}\,\ell_{t-}\Big(\lambda_t+\frac{z_t}{\ell_{t-}}\Big)^\top\,d\br{M}_t\, \Big(\lambda_t+\frac{z_t}{\ell_{t-}}\Big) + z_t\,dM_t + dn_t
\end{align}
with terminal condition $\ell_T = \exp(B)$, and the optimal strategy $\hvartheta$ is
\begin{equation}\label{eq:ExpOptimalStrategy}
  \hvartheta=\lambda + \frac{Z^{L^{\exp}}}{L^{\exp}_-}.
\end{equation}
This can be derived directly by dynamic programming or inferred, e.g., from Frei and Schweizer~\cite[Proposition~1]{FreiSchweizer.09}. We will actually reprove the BSDE later,
but present it already at this stage for the following motivation.

We observe that~\eqref{eq:BellmanBsdeExp} \emph{coincides with the BSDE~\eqref{eq:BellmanBSDEforL},
except that $q$ is replaced by $1$} and the terminal condition is $\exp(B)$ instead of $D_T$. From now on \emph{we assume
$\exp(B)=D_T$}, then one can guess that the solutions $L(p)$ should converge to $L^{\exp}$ as $q\to1-$, or equivalently $p\to-\infty$.

\begin{Thm}\label{th:strategyConvExponential}
  Let $S$ be continuous.
  \begin{enumerate}[topsep=3pt, partopsep=0pt, itemsep=1pt,parsep=2pt]
  \item  As $p\downarrow -\infty$, $L_t(p)\downarrow L_t^{\exp}$ $P$-a.s.~for all $t$, with a uniform bound.

  \item If $L(p)$ is continuous for each $p<0$, then $L^{\exp}$ is also continuous and the convergence $L(p)\downarrow L^{\exp}$ is uniform in $t$, $P$-a.s. Moreover,
      \[
        (1-p)\, \hpi(p) \to \hvartheta\quad \mbox{ in } L^2_{loc}(M).
      \]
  \end{enumerate}
\end{Thm}

We note that (ii) is also a statement about the rate of convergence for $\hpi(p)\to0$ in Theorem~\ref{th:Limit-inftyEconomics}(ii) for the case without intermediate consumption.
The proof occupies most of the remainder of the section. Part~(i) follows from the next two lemmata;
recall that the monotonicity of $p\mapsto L_t(p)$ was already established in Proposition~\ref{pr:ComparisonLNoCons} while the uniform bound is from Lemma~\ref{le:BoundsForL}.

\begin{Lemma}\label{le:LpGreaterLexp}
  We have $L(p)\geq L^{\exp}$ for all $p<0$.
\end{Lemma}

\begin{proof}
  As is well-known, we may assume that $B=0$ by a change of measure from  $P$ to $dP(B)=(e^B/E[e^B])\,dP$.
  Let $Q^E\in \sM^{ent}$ be the measure with minimal entropy $H(Q|P)$; see, e.g.,~\cite[Theorem~3.5]{KabanovStricker.02}.
  Let $Y$ be its $P$-density process, then
  \begin{equation}\label{eq:dualRelationForExp}
    -\log(L^{\exp}_t)=E^{Q^E}\big[\log(Y_T/Y_t)\big|\cF_t\big]=E\big[(Y_T/Y_t)\log(Y_T/Y_t)\big|\cF_t\big].
  \end{equation}
  This is merely a
  dynamic version of the well-known duality relation
  stated, e.g., in~\cite[Theorem~2.1]{KabanovStricker.02} %
  and one can retrieve this version, e.g., from~\cite[Eq.~(8),(10)]{FreiSchweizer.09}.
  Using the decreasing function $\phi$ from Lemma~\ref{le:monotoneFunctionRevHolder},
  \begin{align*}
    L^{\exp}_t&=\exp\Big(-E\big[(Y_T/Y_t)\log(Y_T/Y_t)\big|\cF_t\big]\Big)=\phi(1)\\
      & \leq \phi(q)= E\big[(Y_T/Y_t)^{q}\big|\cF_t\big]^{1/\beta}\leq L^*(p)^{1/\beta}=L(p),
  \end{align*}
  where \eqref{eq:DualOppProcAltDef} was used for the second inequality.
\end{proof}

\begin{Lemma}
  Let $S$ be continuous. Then $\limsup_{p\to -\infty} L_t(p)\leq L_t^{\exp}$.
\end{Lemma}

\begin{proof}
  Fix $t\in[0,T]$. We denote $\cE_{tT}(X):=\cE(X)_T/\cE(X)_t$ and $X_{tT}:=X_T-X_t$.

  (i)~~Let $\vartheta\in L(R)$ be such that $|\vartheta \sint R| + \br{\vartheta\sint R}$ is bounded by a constant.
   Noting that $L(R)\subseteq \cA$ because $R$ is continuous, we have from~\eqref{eq:OppProcIndep} that
   \begin{align*}
    L_t(p) & = \essinf_{\pi\in\cA} E\big[D_T\cE^p_{tT}(\pi\sint R)\big|\cF_t\big]
                  \leq E\big[D_T\cE^p_{tT}(|p|^{-1}\vartheta \sint R)\big|\cF_t\big]\\
           & = E\big[D_T\exp\big(-(\vartheta \sint R)_{tT} + \tfrac{1}{2|p|} \br{\vartheta\sint R}_{tT}\big) \big|\cF_t\big].
  \end{align*}
  The expression under the last conditional expectation is bounded uniformly in $p$, so the last line
  converges to $E\big[\exp\big(B-(\vartheta \sint R)_{tT}\big) \big|\cF_t\big]$ $P$-a.s.\ when $p\to-\infty$;
  recall $D_T=\exp(B)$. We have shown
  \begin{equation}\label{eq:proofExpConvBddStrat}
    \limsup_{p\to-\infty} L_t(p) \leq E\big[\exp\big(B-(\vartheta \sint R)_{tT}\big) \big|\cF_t\big]\quad P\mbox{-a.s.}
  \end{equation}

  (ii)~~Let $\vartheta\in L(R)$ be such that $\exp(-\vartheta\sint R)$ is of class (D). Defining the
  stopping times $\tau_n=\inf\{s>0:\,|\vartheta \sint R_{s}| + \br{\vartheta\sint R}_{s} \geq n\}$, we have
  \[
    \limsup_{p\to-\infty} L_t(p) \leq E\big[\exp\big(B-(\vartheta \sint R)^{\tau_n}_{tT}\big) \big|\cF_t\big]\quad P\mbox{-a.s.}
  \]
  for each $n$, by step (i) applied to $\vartheta1_{(0,\tau_n]}$.
  Using the class (D) property, the right hand side converges to $E\big[\exp\big(B-(\vartheta \sint R)_{tT}\big) \big|\cF_t\big]$ in $L^1(P)$ as $n\to\infty$, and also $P$-a.s.\ along a subsequence. Hence~\eqref{eq:proofExpConvBddStrat} again holds.

  (iii)~~The previous step has a trivial extension: Let $g_{tT}\in L^0(\cF_T)$ be a random variable such that
  $g_{tT}\leq (\vartheta \sint R)_{tT}$ for some $\vartheta$ as in (ii). Then
  \[
    \limsup_{p\to-\infty} L_t(p) \leq E\big[\exp(B-g_{tT}) \big|\cF_t\big]\quad P\mbox{-a.s.}
  \]

  (iv)~~Let $\hvartheta\in\Theta$ be the optimal strategy. We claim that there exists a sequence $g^n_{tT}\in L^0(\cF_T)$
  of random variables as in (iii) such that
  \[
    \exp\big(B-g^n_{tT}\big)\to \exp\big(B-G_{t,T}(\hvartheta)\big)\;\mbox{ in }L^1(P).
  \]
  Indeed, we may assume $B=0$, as in the previous proof.
  Then our claim follows by the construction of Schachermayer~\cite[Theorem~2.2]{Schachermayer.01} applied to
  the time interval $[t,T]$; recall the definitions~\cite[Eq.\,(4),(5)]{Schachermayer.01}.
  We conclude that $\limsup_{p\to-\infty} L_t(p) \leq E\big[\exp\big(B-G_{t,T}(\hvartheta)\big) \big|\cF_t\big]=L^{\exp}_t$
  $P$-a.s.\ by the $L^1(P)$-continuity of the conditional expectation.
\end{proof}

\begin{Remark}
  Recall that $\exp(-G(\hvartheta))L^{\exp}$ is a martingale, hence of class~(D). If $L^{\exp}$ is uniformly bounded away from zero, it follows that $\exp(-G(\hvartheta))$ is already of class (D) and the last two steps in the previous proof are unnecessary.
  This situation occurs precisely when the right hand side of~\eqref{eq:dualRelationForExp} is bounded uniformly in $t$.
  In standard terminology, the latter condition states that the reverse H\"older inequality $\textrm{R}_{L\log(L)}(P)$
  is satisfied by the density process of the minimal entropy martingale measure.
\end{Remark}

\begin{Lemma}\label{le:KobylanskiConv}
  Let $S$ be continuous and assume that $L(p)$ is continuous for all $p<0$. Then
  $L^{\exp}$ is continuous and $L_t(p)\to L_t^{\exp}$ uniformly in $t$, $P$-a.s. Moreover,
  $Z^{L(p)}\to Z^{\exp}$ in $L^2_{loc}(M)$ and $N(p)\to N^{\exp}$
  in $\cH^2_{loc}$.
\end{Lemma}

We have already identified the monotone limit $L^{\exp}_t=\lim L_t(p)$. Hence, by uniqueness of the KW decomposition, the above
lemma follows from the subsequent one, which we state separately to clarify the argument.
The most important input from the control problems is that by stopping, we can bound $L(p)$ away from zero simultaneously for all $p$ (cf.~Lemma~\ref{le:uniformStopping-infty}).

\begin{Lemma}\label{le:KobylanskiConv2}
  Let $S$ be continuous and assume that $L(p)$ is continuous for all $p<0$. Then
  $\big(L(p),Z^{L(p)},N(p)\big)$ converge to a solution $(\tL, \widetilde{Z},\tN)$
  of the BSDE~\eqref{eq:BellmanBsdeExp} as $p\to -\infty$: $\tL$ is continuous and $L_t(p)\to\tL_t$ uniformly in $t$, $P$-a.s.;
  while $Z^{L(p)}\to \widetilde{Z}$ in $L^2_{loc}(M)$ and $N(p)\to\tN$
  in $\cH^2_{loc}$.
\end{Lemma}

\begin{proof}
  For notational simplicity, we write the proof for the one-dimensional case ($d=1$).
  We fix a sequence $p_n\downarrow -\infty$ and corresponding $q_n\uparrow 1$.
  As $p\mapsto L_t(p)$ is monotone and positive, the $P$-a.s.~limit $\tL_t:=\lim_n L_t(p_n)$ exists.

  The sequence $M^{L(p_n)}$ of martingales is bounded in
  the Hilbert space $\cH^2$ by Lemma~\ref{le:LisBMO}(i). Hence it has a subsequence, still denoted by $M^{L(p_n)}$, which converges
  to some $\tM\in \cH^2$ in the weak topology of $\cH^2$. If we denote the KW decomposition by
  $\tM=\tZ\sint M + \tN$, we have by orthogonality that $Z^{L(p_n)}\to\tZ$ weakly in $L^2(M)$ and $N^{L(p_n)}\to \tN$ weakly in $\cH^2$. We shall use the BSDE to deduce a strong convergence.

  The drivers in the BSDE~\eqref{eq:BellmanBSDEforL} corresponding to $p_n$ and in~\eqref{eq:BellmanBsdeExp} are
  \[
     f^n(t,l,z) := q_n\,f(t,l,z), \quad f(t,l,z):=\frac{1}{2}\,l\Big(\lambda_t+\frac{z}{l}\Big)^2
  \]
  for $(t,l,z)\in [0,T]\times (0,\infty)\times \R$.
  For fixed $t$ and any convergent sequence $(l_m,z_m)\to(l,z)\in (0,\infty)\times \R$, we have
  \[
    f^m(t,l_m,z_m)\to f(t,l,z)\;\; \mbox{$P$-a.s.}
  \]
  By Lemmata~\ref{le:LpGreaterLexp} and~\ref{le:BoundsForLexp} we can find a localizing sequence $(\tau_k)$ such that
  \[
    1/k< L(p)^{\tau_k}\leq k_2 \quad\mbox{for all } p<0,
  \]
  where the upper bound is from Lemma~\ref{le:BoundsForL}. For the processes from~\eqref{eq:StructureContForR} we may assume that $\lambda^{\tau_k}\in L^2(M)$ and $M^{\tau_k}\in \cH^2$ for each $k$.

  To relax the notation, let $L^n=L(p_n)$, $Z^n=Z^{L(p_n)}$, $N^n=N^{L(p_n)}$, and $M^n=M^{L(p_n)}=Z^n\sint M + N^n$.
  The purpose of the localization is that $(f^n)$ are
  uniformly quadratic in the relevant domain: As $(L^n,Z^n)^{\tau_k}$ takes values in $[1/k,k_2]\times \R$ and
  \[
    |f^n(t,l,z)|\leq \big|l\lambda_t^2 +\lambda_t z + z^2/l\big|\leq (1+l)\lambda_t^2 + (1+1/l)z^2,
  \]
  we have for all $m,n\in\N$ that
  \begin{align}\label{eq:quadEstimate}
    |f^m(t&,L_t^n,Z_t^n)|^{\tau_k} \leq \xi_t + C_k(Z_{t\wedge \tau_k}^n)^2, \quad\mbox{where}\\
    &\xi:=(1+k_2)\big(\lambda^{\tau_k}\big)^2 \in L^1_{\tau_k}(M),\quad C_k:=1+k.\nonumber
  \end{align}
  Here $L^r_{\tau}(M):=\{H\in L^2_{loc}(M):\, H1_{[0,\tau]}\in L^r(M)\}$
  for a stopping time $\tau$ and $r\geq1$. Similarly, we set $\cH^2_{\tau}=\{X\in\cS:\, X^\tau\in \cH^2\}$.
  Now the following can be shown using a technique of Kobylanski~\cite{Kobylanski.00}.

  \begin{Lemma}\label{le:KobylanskiConv3} For fixed $k$,
   \begin{enumerate}[topsep=3pt, partopsep=0pt, itemsep=1pt,parsep=2pt]
    \item $Z^n\to\tZ$ in $L_{\tau_k}^2(M)$ and $N^n\to \tN$ in $\cH^2_{\tau_k}$,
    \item $\sup_{t\leq T}|L^n_{t\wedge \tau_k}-\tL_{t\wedge \tau_k}|\to 0$ $P$-a.s.
   \end{enumerate}
  \end{Lemma}
  The proof is deferred to Appendix~\ref{se:KobylanskiArgument}.
  Since~(ii) holds for all $k$, it follows that $\tL$ is continuous. Now Dini's Lemma shows $\sup_{t\leq T} |L^n_t-\tL_t|\to 0$ $P$-a.s.\ as claimed. Lemma~\ref{le:KobylanskiConv3} also implies that the limit $(\tL,\tZ,\tN)$ satisfies the BSDE~\eqref{eq:BellmanBsdeExp} on $[0,\tau_k]$ for all $k$, hence on $[0,T]$. The terminal condition is satisfied as $L^n_T=D_T=\exp(B)$ for all $n$.

  To end the proof, note that the convergences hold for the original sequence $(p_n)$, rather than just for a subsequence, since $p\mapsto L(p)$ is monotone and since our choice of $(\tau_k)$ does not depend on the subsequence.
\end{proof}

We can now finish the proof of Theorem~\ref{th:strategyConvExponential} (and Theorem~\ref{th:strategyConvExponentialMainRes}).

\begin{proof}[Proof of Theorem~\ref{th:strategyConvExponential}.]
  Part~(i) was already proved. For~(ii), uniform convergence and continuity were shown in Lemma~\ref{le:KobylanskiConv}.
  In view of~\eqref{eq:optStrategy} and~\eqref{eq:ExpOptimalStrategy}, the claim for the strategies is that
  \[
    (1-p)\, \hpi(p)=\lambda+\frac{Z^{L(p)}}{L(p)}\to \lambda + \frac{Z^{L^{\exp}}}{L^{\exp}}=\hvartheta\quad \mbox{in }L^2_{loc}(M).
  \]
  By a localization as in the previous proof, we may assume that $L(p) + (L(p))^{-1}+ L^{\exp}+(L^{\exp})^{-1}$
  is bounded uniformly in $p$, and, by Lemma~\ref{le:KobylanskiConv}, that
  $Z^{L(p)}\sint M \to Z^{\exp}\sint M$ in $\cH^2$. We have
  \begin{align*}
   \Big\| \tfrac{Z^{L(p)}}{L(p)}  &\sint M - \tfrac{Z^{L^{\exp}}}{L^{\exp}} \sint M \Big\|_{\cH^2}  \\
      & \leq \Big\| \tfrac{1}{L(p)} \big(Z^{L(p)}-Z^{L^{\exp}}\big) \sint M \Big\|_{\cH^2}
      + \Big\| \Big(\tfrac{1}{L(p)}-\tfrac{1}{L^{\exp}}\Big) Z^{L^{\exp}} \sint M \Big\|_{\cH^2} .
  \end{align*}
  Clearly the first norm converges to zero. Noting that  $Z^{L^{\exp}} \sint M\in \cH^2$ (even $BMO$) due to Lemma~\ref{le:QuadVarLemma},
  the second norm tends to zero by dominated convergence for stochastic integrals.
\end{proof}

The last result of this section concerns the convergence of the (normalized) solution $\hY(p)$ of the dual problem~\eqref{eq:dualProblem}; see also the comment after Remark~\ref{rk:heuristicsExpConv}. We recall the assumption~\eqref{eq:finiteEntropy} and that there is no intermediate consumption. To state the result, let
$Q^E(B)\in\sM$ be the measure which minimizes the relative entropy $H(\,\cdot\,|P(B))$ over $\sM$, where $dP(B):=(e^B/E[e^B])\,dP$.
For $B=0$ this is simply the minimal entropy martingale measure, and the existence of $Q^E(B)$ follows from the latter by a change of measure.

\begin{Prop}\label{pr:dualConvExp}
  Let $S$ be continuous and assume that $L(p)$ is continuous for all $p<0$. Then
  $\hY(p)/\hY_0(p)$ converges in the semimartingale topology to the density process of $Q^E(B)$
  as $p\to-\infty$.
\end{Prop}

\begin{proof}
  We deduce from Lemma~\ref{le:KobylanskiConv} that $L^{-1}\sint N\to (L^{\exp})^{-1}\sint N^{\exp}$ in $\cH^2_{loc}$, as in the previous proof.
  Since $\hY/\hY_0=\cE(-\lambda\sint M + L^{-1}\sint N)$ by~\eqref{eq:DualOptimizerFormula}, Lemma~\ref{le:convStochExp}(ii) shows that $\hY/\hY_0\to \cE\big(-\lambda\sint M + (L^{\exp})^{-1}\sint N^{\exp}\big)$
  in the semimartingale topology. The right hand side is the density process of $Q^E(B)$; this follows, e.g., from~\cite[Proposition~1]{FreiSchweizer.09}.
\end{proof}

\section{The Limit $p\to 0$}\label{se:pto0Proofs}

In this section we prove Theorem~\ref{th:Limit0Economics}, some refinements of that result, as well as the corresponding convergence for the opportunity processes and the dual problem. Due to substantial technical differences, we consider separately the limits $p\to0$ from below and from above.
Recall the semimartingale
$\eta_t=E\big[\int_t^T D_s\,\mu^\circ(ds)\big|\cF_t\big]$ with canonical decomposition
\begin{equation}\label{eq:etaDecomp}
  \eta_t=(\eta_0+M^\eta_t)+A^\eta_t=E\Big[\int_0^T D_s\,\mu^\circ(ds)\Big|\cF_t\Big] - \int_0^t D_s\,\mu(ds).
\end{equation}
Clearly $\eta$ is a supermartingale with continuous finite variation part, and a martingale in the case without intermediate consumption ($\mu=0$). From~\eqref{eq:BoundsR} we have the uniform bounds
\begin{equation}\label{eq:etaBounds}
    0< k_1\leq \eta \leq (1+T)k_2.
\end{equation}

\subsection{The Limit $p\to 0-$}
We start with the convergence of the opportunity processes.

\begin{Prop}\label{pr:DualLconv0-} As $p\to 0-$,
  \begin{enumerate}[topsep=3pt, partopsep=0pt, itemsep=1pt,parsep=2pt]
   \item  for each $t\in[0,T]$, $L^*_t(p)\to \eta_t$ $P$-a.s.\ and in $L^r(P)$ for $r\in [1,\infty)$, with a uniform bound.

   \item if $\F$ is continuous, then $L^*_t(p)\to \eta_t$ uniformly in $t$, $P$-a.s.; and in $\cR^r$ for $r\in[1,\infty)$.

   \item if $u_{p_0}(x_0)<\infty$ for some $p_0\in(0,1)$, then $L^*_t(p)\to \eta_t$ uniformly in $t$, $P$-a.s.; in $\cR^r$ for $r\in[1,\infty)$; and prelocally in $\cR^\infty$.

  \end{enumerate}
  The same assertions hold for $L^*$ replaced by $L$.
\end{Prop}

\begin{proof}
  We note that $p\to 0-$ implies $q\to0+$ and $\beta\to 1-$. In view of $L=(L^*)^{1/\beta}$, it suffices to prove the claims for $L^*$.
  From Lemma~\ref{le:BoundsForL},
  \begin{equation}\label{eq:unifBoundProofDualLconv0-}
    0\leq L^*_t(p) \leq \mu^\circ[t,T]^{-\beta p}\,E\Big[\int_t^T D_s\, \mu^\circ(ds)\Big|\cF_t\Big]^\beta\to \eta_t\;\mbox{ in }\cR^\infty.
  \end{equation}
  To obtain a lower bound, we consider the density process $Y$ of some $Q\in\sM$.

  (i)~~Using \eqref{eq:DualOppProcAltDef} we obtain
  \begin{align*}
    L^*_t(p) %
     \geq  \int_t^T  E\big[D_s^\beta (Y_s/Y_t)^{q}\big|\cF_t\big]\,\mu^\circ(ds).
  \end{align*}
  Clearly $D_s^\beta\to D_s$ in $\cR^\infty$ and $(Y_s/Y_t)^{q}\to 1$ $P$-a.s.\ for $q\to0$. We can argue as
  in Proposition~\ref{pr:DualLconv-infty}: For $s\geq t$ fixed,
  $0\leq (Y_s/Y_t)^{q}\leq 1+Y_s/Y_t\in L^1(P)$ yields $E\big[D_s^\beta(Y_s/Y_t)^{q}\big|\cF_t\big]\to E[D_s|\cF_t]$ $P$-a.s.
  Since $Y^q$ is a supermartingale, $0\leq E\big[D_s^\beta (Y_s/Y_t)^{q}\big|\cF_t\big] \leq 1\vee k_2$, and we conclude for each $t$ that
  \[
    \int_t^T E\big[D_s^\beta (Y_s/Y_t)^{q}\big|\cF_t\big]\,\mu^\circ(ds)
    \to \int_t^T E\big[D_s \big|\cF_t\big]\,\mu^\circ(ds)= \eta_t\;\;P\mbox{-a.s.}
  \]
  Hence $L^*_t(p)\to \eta_t$ $P$-a.s.\ and the convergence in $L^r(P)$ follows by the bound~\eqref{eq:unifBoundProofDualLconv0-}.

  (ii)~~Assume that $\F$ is continuous. Our argument will be similar to Proposition~\ref{pr:DualLconv-infty}, but the source of monotonicity is different.
  Fix $(s,\omega)\in[0,T]\times\Omega$ and consider
  \[
    g_q(t):=E\big[(Y_s/Y_t)^{q}\big|\cF_t\big]^{\frac{1}{1-q}}(\omega), \quad t\in[0,s].
  \]
  Then $g_q(t)$ is continuous in $t$ and decreasing in $q$ by virtue of Lemma~\ref{le:monotoneFunctionRevHolder}.
  Dini's lemma yields $g_q\to 1$ uniformly on $[0,s]$, hence \mbox{$E\big[(Y_s/Y_t)^{q}\big|\cF_t\big]\to1$} uniformly in $t$.
  We deduce that $E\big[D_s^\beta(Y_s/Y_t)^{q}\big|\cF_t\big](\omega) \to E[D_s|\cF_t](\omega)$ uniformly in $t$ since
  \begin{align*}
    \Big|& E\big[D_s^\beta(Y_s/Y_t)^{q}\big|\cF_t\big] - E[D_s|\cF_t] \Big|\\
      &\leq E\big[|D_s^\beta-D_s| (Y_s/Y_t)^{q}\big|\cF_t\big] + \Big|E\big[D_s\{(Y_s/Y_t)^{q}-1\}\big|\cF_t\big]\Big|\\
      &\leq \|D_s^\beta-D_s\|_{L^\infty(P)} E\big[ (Y_s/Y_t)^{q}\big|\cF_t\big]
        + \|D_s\|_{L^\infty(P)} \Big|E\big[ (Y_s/Y_t)^{q}\big|\cF_t\big]-1\Big|\\
      &\leq \|D_s^\beta-D_s\|_{L^\infty(P)}
        + k_2 \Big|E\big[ (Y_s/Y_t)^{q}\big|\cF_t\big]-1\Big|.
  \end{align*}
  The rest of the argument is like the end of the proof of Proposition~\ref{pr:DualLconv-infty}.

  (iii)~~Let $u_{p_0}(x_0)<\infty$ for some $p_0\in(0,1)$. Then we can take
  a different approach via Proposition~\ref{pr:ComparisonDualL}, which shows that
  \[
   L_t^*(p)\;\geq\;E\Big[\int_t^T D_s^{\beta}\,\mu^\circ(ds)\Big|\cF_t\Big]^{1-q/q_0} \,\Big(k_1^{\beta-\beta_0} L^*_t(p_0)\Big)^{q/q_0}
  \]
  for all $p<0$, where we note that $q_0<0$. Using that almost every path of $L^*(p_0)$ is
  bounded and bounded away from zero (Lemma~\ref{le:BoundsForL}), the right hand side $P$-a.s.\ tends to
  $\eta_t=E[\int_t^T D_s\,\mu^\circ(ds)|\cF_t]$ uniformly in $t$ as $q\to0$.
  Since $L^*(p_0)$ is prelocally bounded, the prelocal convergence in
  $\cR^\infty$ follows in the same way.
\end{proof}

\begin{Remark}\label{rk:RinftyConvForL}
  One can ask when the convergence in Proposition~\ref{pr:DualLconv0-} holds even in $\cR^\infty$. The following statements remain valid if $L^*$ replaced by $L$.
  \begin{enumerate}[topsep=3pt, partopsep=0pt, itemsep=1pt,parsep=2pt]
   \item Assume again that $u_{p_0}(x_0)<\infty$ for some $p_0\in(0,1)$, and in addition that
      $L^*(p_0)$ is (locally) bounded. Then the argument for Proposition~\ref{pr:DualLconv0-}(iii) shows $L^*(p)\to \eta$ in $\cR^\infty$ ($\cR^\infty_{loc}$).
   \item Conversely, $L^*(p)\to \eta$ in $\cR^\infty$ ($\cR^\infty_{loc}$) implies that $L^*(p)$ is (locally) bounded away from zero for all $p<0$ close to zero, because $\eta\geq k_1>0$.
  \end{enumerate}
\end{Remark}

As we turn to the convergence of the martingale part $M^{L(p)}$, a suitable localization will again be crucial.

\begin{Lemma}\label{le:UniformStoppingLimit0-}
  Let $p_1<0$. There exists a localizing sequence $(\sigma_n)$ such that
  \[
    \big(L(p)\big)^{\sigma_n}_->1/n\;\; \mbox{simultaneously for all }p\in [p_1,0).
  \]
\end{Lemma}

\begin{proof}
  This follows from Proposition~\ref{pr:ComparisonLNoCons} and Lemma~\ref{le:BoundsForL}.
\end{proof}

Next, we state a basic result (i) for the convergence of $M^{L(p)}$ in $\cH^2_{loc}$
and stronger convergences under additional assumptions (ii) and (iii), for which Remark~\ref{rk:RinftyConvForL}(i) gives sufficient conditions.

\newpage

\begin{Prop}\label{pr:Strategy0-}
  Assume that $S$ is continuous. As $p\to 0-$,
  \begin{enumerate}[topsep=3pt, partopsep=0pt, itemsep=1pt,parsep=2pt]
    \item $M^{L(p)}\to M^{\eta}$ in $\cH^2_{loc}$.
    \item if $L(p)\to \eta$ in $\cR^\infty_{loc}$, then $M^{L(p)}\to M^{\eta}$ in $BMO_{loc}$.
    \item if $L(p)\to \eta$ in $\cR^\infty$, then $M^{L(p)}\to M^{\eta}$ in $BMO$.
  \end{enumerate}
\end{Prop}

\begin{proof}
  Set $X=X(p)=\eta-L(p)$. Then $X$ is bounded uniformly in $p$ by Lemma~\ref{le:BoundsForL} and our aim is
  to prove $M^{X(p)}\to 0$.
  Lemma~\ref{le:QuadVarLemma} applied to $\|\eta\|_{\infty} -\eta$ shows that $M^\eta\in BMO$.
  We may restrict our attention to $p$ in some interval $[p_1,0)$ and Lemma~\ref{le:LisBMO} shows that
  $\sup_{p\in [p_1,0)}  \|M^{L(p)}\|_{BMO}<\infty$. Due to the orthogonality of the sum $M^L=Z^L\sint M +N^L$, we have in particular that
  \begin{equation}\label{eq:proofOfStrategy0-}
    \sup_{p\in [p_1,0)}  \|Z^L(p)\sint M\|_{BMO}<\infty.
  \end{equation}
  Under the condition of (iii), $L(p)$ is bounded away from zero for all $p$ close to zero since $\eta\geq k_1>0$; moreover,
  $\lambda\sint M \in BMO$ by Corollary~\ref{co:lambdaBMO}. For (i) and (ii) we may assume by a localization as in Lemma~\ref{le:UniformStoppingLimit0-} that $L_-(p)$ is bounded away from zero uniformly in $p$. Since $M$ is continuous, we may also assume that $\lambda\sint M \in BMO$, by another localization.

  Using the formula~\eqref{eq:FVpartofL} for $A^L$ and the decomposition~\eqref{eq:etaDecomp} of $\eta$, the finite variation part $A^X$ is continuous and
  \begin{align}\label{eq:formulaAXproof}
    2\,dA^X &
    = 2\Big\{(1-p)D^\beta L_-^q - D\Big\}\,d\mu \\
    &\phantom{=} - q\Big\{L_- \lambda^\top\,d\br{M}\,\lambda  + 2\lambda^\top\,d\br{M}\, Z^L + L_-^{-1} \big(Z^L\big)^\top \,d\br{M}\, Z^L\Big\}.\nonumber
  \end{align}
  In particular, we note that
  \begin{equation}\label{eq:MXbracket}
    [M^X]=[X]-X_0^2 = X^2-X_0^2 - 2\int X_-\,dX.
  \end{equation}
  For case (i) we have $X_0^2\to 0$ and $E[X_T^2]\to 0$ by Proposition~\ref{pr:DualLconv0-} (Remark~\ref{rk:convAtStoppingtime} applies). In case (iii) we have $X\to0$ in $\cR^\infty$ by assumption and under (ii) the same holds after a localization.
  If we denote $o^1_t:=E\big[X^2_T - X^2_t \big|\cF_t\big]$, we therefore have that $o^1_0\to0$ in case (i) and
  $o^1\to 0$ in $\cR^\infty$ in cases (ii) and (iii). Denote also
  $o^2_t:=2E\big[\int_t^T X_-\{(1-p)D^\beta L_-^q - D\}\,d\mu\big|\cF_t\big]$.
  Recalling that $p\to 0-$ implies $q\to0+$ and $\beta\to 1-$, we have $(1-p)D^\beta L_-^q - D\to 0$ in $\cR^\infty$
  and since $X_-$ is bounded uniformly in $p$, it follows that $o^2\to 0$ in $\cR^\infty$.
  As $M^X\in BMO$ and $X_-$ is bounded, $\int X_-\,dM^X$ is a martingale and~\eqref{eq:MXbracket} yields
  \[
    E\big[[M^X]_T - [M^X]_t \big|\cF_t\big] = E\big[X^2_T - X^2_t \big|\cF_t\big] -2 E\Big[\int_t^T X_-\,dA^X \Big|\cF_t\Big].
  \]
  Using~\eqref{eq:formulaAXproof} and the definitions of $o^1$ and $o^2$, we can rewrite this as
  \begin{align*}
    &E\big[[M^X]_T - [M^X]_t \big|\cF_t\big] - o^1_t+o^2_t \\
    &= q E\Big[\int_t^T X_-\big\{L_- \lambda^\top\,d\br{M}\,\lambda  + 2\lambda^\top\,d\br{M}\, Z^L + L_-^{-1} \big(Z^L\big)^\top \,d\br{M}\, Z^L\big\} \Big|\cF_t \Big].
  \end{align*}
  Applying the Cauchy-Schwarz inequality and using that $X_-,L_-, L_-^{-1}$ are bounded uniformly in $p$, it follows that
  \begin{align*}
    E\big[[M^X]_T - [M^X]_t &\big|\cF_t\big] - o^1_t+o^2_t \\
    &\leq q E\Big[\int_t^T X_-(1+L_-)\lambda^\top\,d\br{M}\,\lambda \Big|\cF_t \Big]\\
    &\phantom{=} +q E\Big[\int_t^T X_-(1+L_-^{-1})\big(Z^L\big)^\top \,d\br{M}\, Z^L \Big|\cF_t \Big]\\
    & \leq q C \big(\|\lambda \sint M\|_{BMO}+\|Z^{L(p)} \sint M\|_{BMO}\big),
  \end{align*}
  where $C>0$ is a constant independent of $p$ and $t$. In view of~\eqref{eq:proofOfStrategy0-}, the right hand side is bounded
  by $qC'$ with a constant $C'>0$ and we have
  \[
    E\big[[M^X]_T - [M^X]_{t} \big|\cF_t\big] \leq qC' + o^1_t-o^2_t.
  \]
  For (i) we only have to prove the convergence to zero of the left hand side for $t=0$ and so this ends the proof. For (ii) and (iii)
  we use $[M^X]_{t}=[M^X]_{t-}+(\Delta M^X_t)^2$ and $|\Delta M^X|=|\Delta X|\leq 2\|X\|_{\cR^\infty}$ to obtain
  \[
    \sup_{t\leq T} E\big[[M^X]_T - [M^X]_{t-} \big|\cF_t\big]  \;\leq\;  qC' + \|o^1\|_{\cR^\infty}+\|o^2\|_{\cR^\infty} + 4\|X\|^2_{\cR^\infty}
  \]
  and we have seen that the right hand side tends to $0$ as $p\to 0-$.
\end{proof}

\subsection{The Limit $p\to 0+$}\label{subsec:limit0+}

We notice that the limit of $L(p)$ for $p\to 0+$ is meaningless without supposing that $u_{p_0}(x_0)<\infty$ for some $p_0\in(0,1)$,
so we make this a \emph{standing assumption} for the entire Section~\ref{subsec:limit0+}.
We begin with a result on the integrability of the tail of the sequence.

\begin{Lemma}\label{le:integrablility0+}
 Let $1\leq r<\infty$. There exists a localizing sequence $(\sigma_n)$ such that
  \[
    \mathop{\esssup}_{t\in [0,T],\; p\in (0,p_0/r]} L_{t\wedge\sigma_n}(p) \q\mbox{is in }L^r(P)\mbox{ for all }n.
  \]
\end{Lemma}

\begin{proof}
  Let $p_1=p_0/r$ and $\sigma_n=\inf \{t>0: L_t(p_1)>n\}\wedge T$, then by Corollary~\ref{co:ComparisonAppl}(ii),
  $\sup_t L_{t\wedge \sigma_n} (p_1)\leq n + \Delta L_{\sigma_n}(p_1) \in L^r(P)$. But
  $L(p)\leq C L(p_1)$ by Corollary~\ref{co:ComparisonAppl}(i), so $(\sigma_n)$ already satisfies the
  requirement.
\end{proof}

\begin{Prop}\label{pr:DualLconv0+}
  As $p\to 0+$,
  \[
    L^*(p)\to \eta,
  \]
  uniformly in $t$, $P$-a.s.; in $\cR^r_{loc}$ for $r\in[1,\infty)$; and prelocally in $\cR^\infty$.
  Moreover, the convergence takes place in $\cR^\infty$ (in $\cR^\infty_{loc}$) if and only if $L(p_1)$ is (locally) bounded for some $p_1\in(0,p_0)$.
  The same assertions hold for $L^*$ replaced by $L$.
\end{Prop}

\begin{proof}
  We consider only  $p\in(0,p_0)$ in this proof
  and recall that $p\to 0+$ implies $q\to0-$ and $\beta\to 1-$. Since $L=(L^*)^{1/\beta}$, it suffices to prove the claims for $L^*$.
  Using Lemma~\ref{le:BoundsForL},
  \begin{equation}\label{eq:ProofDualLconv0+Bound1}
     L^*_t(p)\geq \mu^\circ[t,T]^{-\beta p}\,E\Big[\int_t^T D_s\, \mu^\circ(ds)\Big|\cF_t\Big]^\beta\to \eta_t\;\mbox{ in }\cR^\infty.
  \end{equation}
  Conversely, by Proposition~\ref{pr:ComparisonDualL},
  \begin{equation}\label{eq:ProofDualLconv0+Bound2}
   L_t^*(p)\;\leq\;E\Big[\int_t^T D_s^{\beta}\,\mu^\circ(ds)\Big|\cF_t\Big]^{1-q/q_0} \,\Big(k_1^{\beta-\beta_0} L^*_t(p_0)\Big)^{q/q_0}.
  \end{equation}
  Since almost every path of $L^*(p_0)$ is bounded, the right hand side $P$-a.s.\ tends to
  $\eta_t$ uniformly in $t$ as $q\to0-$.
  By localizing $L^*(p_0)$ to be prelocally bounded, the same argument shows the prelocal convergence in
  $\cR^\infty$.

  We have proved that $L^*(p)\to \eta$ uniformly in $t$, $P$-a.s. In view of Lemma~\ref{le:integrablility0+}, the convergence in $\cR^r_{loc}$ follows by dominated convergence.

  For the second claim, note that the ``if'' statement is shown exactly like the prelocal $\cR^\infty$ convergence and the converse holds by boundedness of $\eta$. Of course, if $L(p_1)$ is (locally) bounded for some $p_1\in(0,p_0)$, then in fact $L(p)$ has this property for all $p\in (0,p_1]$, by Corollary~\ref{co:ComparisonAppl}(i).
\end{proof}

We turn to the convergence of the martingale part. The major difficulty will be that $L(p)$ may have unbounded jumps; i.e., we have to prove the convergence of quadratic BSDEs whose solutions are not locally bounded.

\begin{Prop}\label{pr:Strategy0+}
    Assume that $S$ is continuous. As $p\to 0+$,
  \begin{enumerate}[topsep=3pt, partopsep=0pt, itemsep=1pt,parsep=2pt]
    \item $M^{L(p)}\to M^{\eta}$ in $\cH^2_{loc}$.
    \item if there exists $p_1\in(0,p_0]$ such that $L(p_1)$ is locally bounded, then $M^{L(p)}\to M^{\eta}$ in $BMO_{loc}$.
    \item if there exists $p_1\in(0,p_0]$ such that $L(p_1)$ is bounded, then $M^{L(p)}\to M^{\eta}$ in $BMO$.
  \end{enumerate}
\end{Prop}

The following terminology will be useful in the proof.
We say that real numbers $(x_\eps)$ converge to $x$ \emph{linearly} as $\eps\to0$ if
\[
  \limsup_{\eps\to 0+} \tfrac{1}{\eps}|x_\eps - x|<\infty.
\]

\begin{Lemma}\label{le:linearConvergence}
  Let $x_\eps\to x$ linearly and $y_\eps\to y$ linearly. Then
  \begin{enumerate}[topsep=3pt, partopsep=0pt, itemsep=1pt,parsep=2pt]
    \item $\limsup_{\eps\to0} \tfrac{1}{\eps}|x_\eps-y_\eps|<\infty$ if $x=y$,
    \item $x_\eps y_\eps\to xy$ linearly,
    \item if $x>0$ and $\varphi$ is a real function with $\varphi(0)=1$ and differentiable at $0$, then
        $(x_\eps)^{\varphi(\eps)}\to x$ linearly.
  \end{enumerate}
\end{Lemma}
\begin{proof}
  (i) This is immediate from the triangle inequality. (ii) This follows from $|x_\eps y_\eps-xy|\leq |x_\eps||y_\eps-y|+|y||x_\eps-x|$ because convergent sequences are bounded. (iii) Here we use
  \[
    |(x_\eps)^{\varphi(\eps)}-x|\leq |x_\eps||(x_\eps)^{\varphi(\eps)-1} - 1|+ |x_\eps-x|;
  \]
  as $\{x_\eps\}$ is bounded and $x_\eps\to x$ linearly, the question is reduced to the boundedness of $\eps^{-1}|(x_\eps)^{\varphi(\eps)-1} - 1|$.
  Fix $0<\delta_1<x<\delta_2$ and set $\varrho(\delta,\eps):=|\delta^{\varphi(\eps)-1}-1|$.
  For $\eps$ small enough,  $x_\eps\in [\delta_1,\delta_2]$ and then
  \[
   \varrho(\delta_1,\eps)\wedge \varrho(\delta_2,\eps) \;\leq\; |(x_\eps)^{\varphi(\eps)-1} - 1| \;\leq\; \varrho(\delta_1,\eps)\vee \varrho(\delta_2,\eps).
  \]
  For $\delta>0$ we have $\lim_\eps \eps^{-1}|\varrho(\delta,\eps)|= \big|\frac{d}{d\eps}\delta^{\varphi(\eps)}|_{\eps=0}\big| = |\log(\delta)\varphi'(0)|<\infty$. Hence the upper and the lower bound above converge to $0$ linearly.
\end{proof}

\begin{proof}[Proof of Proposition~\ref{pr:Strategy0+}] We first prove (ii) and (iii), i.e,
  we assume that $L(p_1)$ is locally bounded (resp.\ bounded). Recall $L(p)\geq k_1$ from Lemma~\ref{le:BoundsForL}.
  By Corollary~\ref{co:ComparisonAppl}(i) there exists a constant $C>0$ independent of $p$ such that $L(p)\leq C L(p_1)$ for all $p\in (0,p_1]$. Hence $L(p)$ is bounded uniformly in $p\in (0,p_1]$ in the case (iii) and for (ii) this holds after a localization.
  Now Lemma~\ref{le:LisBMO}(ii) implies
  $\sup_{p\in (0,p_1]}  \|M^{L(p)}\|_{BMO}<\infty$ and we can proceed exactly as in the proof of items~(ii) and~(iii) of Proposition~\ref{pr:Strategy0-}.

  (i)~~This case is more difficult because we have to use prelocal bounds and Lemma~\ref{le:LisBMO}(ii) does not apply.
  Again, we want to imitate the proof of Proposition~\ref{pr:Strategy0-}(i), or more precisely, the arguments after~\eqref{eq:MXbracket}.
  We note that for the claimed $\cH^2_{loc}$-convergence those estimates are required only at $t=0$  and so the $BMO$-norms can be replaced by $\cH^2$-norms. Inspecting that proof in detail, we see that we can proceed in the same way once we establish:

  \begin{itemize}
   \item There exists a localizing sequence $(\sigma_n)$ and constants $C_n$ such that for all $n$,
     \begin{enumerate}[topsep=3pt, partopsep=0pt, itemsep=1pt,parsep=2pt]
       \item[(a)] $(H1_{[0,\sigma_n]})\sint M^{L(p)}$ is a martingale for all $H$ predictable and bounded, and all $p\in (0,p_0)$,
       \item[(b)] $\sup_{p\in (0,p_0]} \big(L_-(p) + L_-^{-1}(p)\big)\leq C_n$ on $[0,\sigma_n]$,
       \item[(c)] $\limsup_{p\to0+}\|Z^{L(p)}1_{[0,\sigma_n]}\|_{L^2(M)}\leq C_n$.
     \end{enumerate}
  \end{itemize}

  We may assume by localization that $\lambda \sint M \in \cH^2$.  We now prove (a)-(c); instead of indicating $(\sigma_n)$ explicitly, we write ``by localization\dots'' as usual.

  (a)~~Fix $p\in(0,p_0)$. By Lemma~\ref{le:BoundsForL} and Lemma~\ref{co:ComparisonAppl}(ii), $L=L(p)$ is a supermartingale of class (D).
  Hence its  Doob-Meyer decomposition $L=L_0 + M^L + A^L$
  is such that $A^L$ is decreasing and nonpositive, and $M^L$ is a true martingale. Thus
  \[
    0\leq E[-A^L_T]=E[L_0-L_T]<\infty.
  \]
  After localizing as in Lemma~\ref{le:integrablility0+} (with $r=1$), we have $\sup_t L_t\in L^1(P)$.
  Hence $\sup_t |M^L_t| \leq \sup_t L_t  - A^L_T \in L^1(P)$.
  Now (a) follows by the BDG inequalities exactly as in the proof of Lemma~\ref{le:QuadVarLemma}.

  (b)~~We have $L_-(p)\geq k_1$ by Lemma~\ref{le:BoundsForL}. Conversely, by Corollary~\ref{co:ComparisonAppl}(i),
   $L_-(p)\leq C L_-(p_0)$ for $p\in (0,p_0]$ with some universal constant $C>0$, and $L_-(p_0)$ is locally bounded by left-continuity.

  (c)~~We shall use the rate of convergence obtained for $L(p)$ and the information about $Z^L$ contained in $A^L$ via the Bellman BSDE.
  We may assume by localization that (a) and (b) hold with $\sigma_n$ replaced by $T$. Thus it suffices to show that
    \begin{align*}
    \limsup_{p\to0+}\bigg\|\sqrt{L_-(p)} \lambda+\frac{Z^{L(p)}}{\sqrt{L_-(p)}}\bigg\|_{L^2(M)} <\infty.
  \end{align*}
  Suppressing again $p$ in the notation, (a) and the formula~\eqref{eq:FVpartofL} for $A^L$ imply
  \begin{align*}
   &E[L_0-L_T]  = E[-A^L_T] \\
   &= E\Big[(1-p) \int_0^T D^\beta L_-^q \,d\mu\Big] - \frac{q}{2} E\Big[\int_0^T \,L_-\Big(\lambda+\frac{Z^L}{L_{-}}\Big)^\top\,d\br{M}\, \Big(\lambda+\frac{Z^L}{L_{-}}\Big)\Big].
  \end{align*}
  Recalling that $L_T=D_T$, this yields
  \begin{align*}
    \tfrac{1}{2} \Big\|\sqrt{L_-} \lambda+\frac{Z^L}{\sqrt{L_-}}\Big\|_{L^2(M)}
    &=  \tfrac{1}{2} E\Big[\int_0^T \,L_-\Big(\lambda+ \frac{Z^L}{L_{-}}\Big)^\top\, d\br{M}\, \Big(\lambda+\frac{Z^L}{L_{-}}\Big)\Big]\\
    &= \tfrac{1}{|q|} \bigg(E[L_0-L_T]-E\Big[(1-p) \int_0^T D^\beta L_-^q \,d\mu\Big]\bigg)\\
    &= \tfrac{1}{|q|} \bigg(L_0 - E\Big[D_T+ (1-p) \int_0^T D^\beta L_-^q \,d\mu\Big]\bigg)\\
    &= \tfrac{1}{|q|} (L_0 - \Gamma_0),
  \end{align*}
  where we have set $\Gamma_0=\Gamma_0(p)=E [D_T+ (1-p) \int_0^T D^\beta L_-^q \,d\mu]$.
  We know that both $L_0$ and $\Gamma_0$ converge to $\eta_0=E\big[\int_0^T D_s\,\mu^\circ(ds)\big]$ as $p\to0+$ (and hence $q\to0-$). However, we are asking for the stronger result
  \[
    \limsup_{p\to 0+} \tfrac{1}{|q|}|L_0(p) - \Gamma_0(p)|<\infty.
  \]
  By Lemma~\ref{le:linearConvergence}(i), it suffices to show that $L_0(p)\to\eta_0$ linearly and $\Gamma_0(p)\to \eta_0$ linearly.
  Using $L^*=L^\beta$, inequalities~\eqref{eq:ProofDualLconv0+Bound1} and~\eqref{eq:ProofDualLconv0+Bound2} evaluated at $t=0$ read
  \[
  \mu^\circ[0,T]^{-p} \eta_0 \leq L_0(p) \leq
        E\Big[\int_0^T D_s^{\beta}\,\mu^\circ(ds)\Big]^{1/\beta+p/q_0} \,\Big(k_1^{1-\beta_0/\beta} L_0(p_0)\Big)^{q/q_0}.
  \]
  Recalling the bound~\eqref{eq:BoundsR} for $D$, items (ii) and (iii) of Lemma~\ref{le:linearConvergence} yield that $L_0(p)\to\eta_0$ linearly. The second claim, that $\Gamma_0(p)\to \eta_0$ linearly, follows from the definitions of $\Gamma_0(p)$ and $\eta_0$ using again~\eqref{eq:BoundsR} and the uniform bounds for $L_-$ from~(b). This ends the proof.
\end{proof}

\subsection{Proof of Theorem~\ref{th:Limit0Economics} and Other Consequences}
\begin{Lemma}\label{le:ForLimit0Economics}
  Assume that $S$ is continuous and that there exists $p_0>0$ such that $u_{p_0}(x_0)<\infty$.
  As $p\to 0$,
  \begin{equation}\label{eq:inLemmaForLimit0Economics}
    \frac{Z^{L(p)}}{L_-(p)}\to \frac{Z^\eta}{\eta_-} \mbox{ in }L^2_{loc}(M)
    \quad\mbox{and} \quad \frac{1}{L_-(p)}\sint N(p)\to \frac{1}{\eta_-}\sint N^\eta \mbox{ in }\cH^2_{loc}.
  \end{equation}
  For a sequence $p\to0-$ the convergence $\frac{Z^{L(p)}}{L_-(p)}\to \frac{Z^\eta}{\eta_-}$ in $L^2_{loc}(M)$ holds also without the assumption on $p_0$.
\end{Lemma}

\begin{proof}
  By localization we may assume that $L_-(p)$ is bounded away from zero and infinity, uniformly in $p$ (Lemma~\ref{le:UniformStoppingLimit0-} and Lemma~\ref{le:BoundsForL} and the preceding proof); we also recall~\eqref{eq:etaBounds}. We have
  \[
    \Big|\frac{Z^{L(p)}}{L_-(p)} - \frac{Z^\eta}{\eta_-}\Big| \leq  \Big|\frac{1}{L_-(p)}\big(Z^{L(p)}-Z^\eta\big)\Big| + \Big|\big(\eta_- - L_-(p)\big)\frac{Z^\eta}{L_-(p)\eta_-}\Big|.
  \]
  Let $u_{p_0}(x_0)<\infty$. The first part of~\eqref{eq:inLemmaForLimit0Economics} follows from the $L^2_{loc}(M)$ and prelocal $\cR^{\infty}$ convergences mentioned in Propositions~\ref{pr:Strategy0-},~\ref{pr:Strategy0+} and Propositions~\ref{pr:DualLconv0-},~\ref{pr:DualLconv0+}, respectively.
  The proof of the second part of~\eqref{eq:inLemmaForLimit0Economics} is analogous.

  Now drop the assumption that $u_{p_0}(x_0)<\infty$ and consider a sequence $p_n\to0-$. Then Proposition~\ref{pr:DualLconv0-} only yields $L_t(p_n)\to\eta_t$ $P$-a.s.\ for each $t$,
  rather than the convergence of $L_{t-}(p_n)$ to $\eta_{t-}$.
  Consider the optional set
  $\Lambda:=\bigcap_n\{L_-(p_n)=L(p_n)\}\cap\{\eta=\eta_-\}$. Because $L(p_n)$ and $\eta$ are c\`adl\`ag,
  $\{t:\,(\omega,t)\in \Lambda^c\}\subset[0,T]$ is countable $P$-a.s.\ and as $M$ is continuous is follows that $\int_0^T1_{\Lambda^c}\,d\br{M}=0$ $P$-a.s.
  Now dominated convergence for stochastic integrals yields that
  $
     \{(\eta_- - L_-(p_n))Z^\eta\}\sint M = \{(\eta - L(p_n))1_{\Lambda}Z^\eta\}\sint M \to 0
  $ in $\cH^2_{loc}$ and the rest is as before.
\end{proof}

\begin{proof}[Proof of Theorem~\ref{th:Limit0Economics} and Remark~\ref{rk:Limit0Economics}]
  The convergence of the optimal consumption is contained in Propositions~\ref{pr:DualLconv0-} and~\ref{pr:DualLconv0+} by the formula~\eqref{eq:consumptionFeedback}. The convergence of the portfolios follows from
  Lemma~\ref{le:ForLimit0Economics} in view of~\eqref{eq:optStrategy}.

  For $p\in (0,p_0]$ we have the uniform bound $\hkappa(p)\leq (k_2/k_1)^{\beta_0}$
  by Lemma~\ref{le:BoundsForL} and~\eqref{eq:consumptionFeedback}; while for
  $p\in[p_1,0)$, $\hkappa(p)$ is prelocally uniformly bounded by
  Lemma~\ref{le:UniformStoppingLimit0-} and~\eqref{eq:consumptionFeedback}.
  Hence the convergence of the wealth processes follows from
  Corollary~\ref{co:ForWealthProcConv}(i).
\end{proof}

We complement the convergence in the primal problem by a result for the solution $\hY(p)$ of the dual problem~\eqref{eq:dualProblem}.

\begin{Prop}\label{pr:dualConv0}
  Assume that $S$ is continuous and that there exists $p_0>0$ such that $u_{p_0}(x_0)<\infty$ holds. Moreover, assume that
  there exists $p_1\in(0,p_0]$ such that $L(p_1)$ is locally bounded. As $p\to0$,
  \[
    \hY(p)\to \frac{\eta_0}{x_0}\cE\Big(-\lambda\sint M + \frac{1}{\eta_-}\sint N^\eta\Big)\quad\mbox{in }\cH^r_{loc} \mbox{ for all }r\in[1,\infty).
  \]
  If $\eta$ and $L(p)$ are continuous for $p<0$, the convergence for a limit $p\to0-$ holds in the semimartingale topology without the assumptions on $p_0$ and $p_1$.
\end{Prop}

\begin{proof}
  (i)~~If $L(p_1)$ is locally bounded, then $L(p)\to\eta$ in $\cR^\infty_{loc}$ by
  Remark~\ref{rk:RinftyConvForL} and Proposition~\ref{pr:DualLconv0+}. Moreover,
  $M^{L(p)}\to M^{\eta}$ in $BMO_{loc}$ by Propositions~\ref{pr:Strategy0-} and~\ref{pr:Strategy0+}.
  This implies $N^{L(p)}\to N^{\eta}$ in $BMO_{loc}$ by orthogonality of the KW decompositions. It follows that
  \[
    - \lambda \sint M+\frac{1}{L_-(p)}\sint N^{L(p)}\to - \lambda \sint M+ \frac{1}{\eta_-}\sint N^{\eta}\quad\mbox{ in } BMO_{loc}.
  \]
  This implies that the corresponding stochastic exponentials converge in $\cH^r_{loc}$ for $r\in[1,\infty)$ (see Theorem~3.4 and Remark~3.7(2) in Protter~\cite{Protter.80}). In view of the formula~\eqref{eq:DualOptimizerFormula} for $\hY(p)$, this ends the proof of the first claim.

  (ii)~~Using Lemma~\ref{le:ForLimit0Economics} and Lemma~\ref{le:convStochExp}(ii), the proof of the second claim is similar.
\end{proof}

Note that in the standard case $D\equiv 1$ the normalized limit in Proposition~\ref{pr:dualConv0} is $\cE(-\lambda\sint M)$, i.e., the ``minimal martingale density'' (cf.~\cite{Schweizer.95b}).
We conclude by an additional statement concerning the convergence of the wealth processes in Theorem~\ref{th:Limit0Economics}.

\begin{Prop}
  Let the conditions of Theorem~\ref{th:Limit0Economics}(ii) hold and assume in addition that
  there exists $p_1\in(0,p_0]$ such that $L(p_1)$ is locally bounded. Then the convergence of the wealth processes in
  Theorem~\ref{th:Limit0Economics}(ii) takes place in $\cH^r_{\loc}$ for all $r\in[1,\infty)$.
\end{Prop}

\begin{proof}
  Under the additional assumption, the results of this section yield the convergence of $\hkappa(p)$ in $\cR^\infty_{loc}$
  and the convergence of $\hpi(p)\sint M$ in $BMO_{loc}$ (and hence in $\cH^\omega_{loc}$) by the same formulas as before. Corollary~\ref{co:ForWealthProcConv}(ii) yields the claim.
\end{proof}

\appendix
\section{Convergence of Stochastic Exponentials}\label{se:convExponentials}

This appendix provides some continuity results for stochastic exponentials of continuous semimartingales in an elementary and self-contained way. They are required for the main results of Section~\ref{se:mainResults} because our wealth processes are exponentials.
We also use a result from the (much deeper) theory of $\cH^\omega$-differentials; but this is applied only for refinements of the main results.

\begin{Lemma}\label{le:locUnifBddForExponential}
  Let $X^n=M^n+A^n$, $n\geq1$ be continuous semimartingales with continuous canonical decompositions and assume that $\sum_n \|X^n\|_{\cH^2}<\infty$. Then $M^n$, $[M^n]$ and
  $\int |dA^n|$ are locally bounded uniformly in $n$.
\end{Lemma}

\begin{proof}
  Let $\sigma_k=\inf \{t>0:\, \sup_n |M^n_t|>k\}\wedge T$. We use the notation $M^{n\star}_t=\sup_{s\leq t}|M^n_s|$,
  then the norms $\|M^{n\star}_T\|_{L^2}$ and $\|M^n\|_{\cH^2}$ are equivalent by the BDG inequalities. Now
  \[
    P\Big[ \sup_n M^{n\star}_T > k\Big]\leq k^{-2} \sum_n \|M^{n\star}_T\|_{L^2}
  \]
  shows $P[\sigma_k<T]\to0$.
  Similarly, $P\big[ \sup_n [M^n]_T > k\big]\leq k^{-1} \sum_n \|M^n\|_{\cH^2}$
  and $P\big[ \sup_n \int_0^T |dA^n| > k\big]\leq k^{-2} \sum_n \|A^n\|_{\cH^2}$
  yield the other claims.
\end{proof}

We sometimes write ``in $\cS^0$'' to indicate convergence in the semimartingale topology.

\begin{Lemma}\label{le:convStochExp}
  Let $X^n=M^n+A^n$, $n\geq1$ and $X=M+A$ be continuous semimartingales with continuous canonical decompositions.
  \begin{enumerate}[topsep=3pt, partopsep=0pt, itemsep=1pt,parsep=2pt]
    \item $\sum_n \|X^n-X\|_{\cH^2}<\infty$ implies $\cE(X^n)\to \cE(X)$ in $\cH^2_{loc}$.
    \item $X^n\to X$ in $\cH^2_{loc}$ implies $\cE(X^n)\to \cE(X)$ in $\cS^0$.
    \item $X^n\to X$ in $\cS^0$ implies $\cE(X^n)\to \cE(X)$ in $\cS^0$.
  \end{enumerate}
\end{Lemma}

\begin{proof}
  (i)~~By localization we may assume that $M$ and $\int |dA|$ are bounded and, by Lemma~\ref{le:locUnifBddForExponential}, that
  $|M^n|$ and $\int |dA^n|$ are bounded by a constant $C$ independent of $n$.
  Note that $X^n\to X$ in $\cH^2$; we shall show $\cE(X^n)\to \cE(X)$ in $\cH^2$. Since this is a metric space, no loss of generality is entailed by passing to a subsequence. Doing so,
  we have $M^n\to M$, $[M^n]\to[M]$, and $A^n\to A$ uniformly in time, $P$-a.s. In view of the uniform bound
  \[
    Y^n:=\cE(X^n)=\exp\big(X^n-\tfrac{1}{2}[M^n]\big)\leq e^{2C}
  \]
  we conclude that $Y^n\to Y:=\cE(X)=\exp(X-\tfrac{1}{2}[M])$ in $\cR^2$.
  By definition of the stochastic exponential we have $Y-Y^n=Y\sint X - Y^n\sint X^n$, where
  \begin{align*}
   \|Y\sint X - Y^n\sint X^n\|_{\cH^2}
     & \leq \|(Y-Y^n)\sint X\|_{\cH^2} + \|Y^n \sint (X-X^n)\|_{\cH^2}.
  \end{align*}
  The first norm tends to zero by dominated convergence for stochastic integrals
  and for the second we use
  that $|Y^n|\leq e^{2C}$ and $X^n\to X$ in $\cH^2$.

  (ii)~~Consider a subsequence of $(X^n)$. After passing to another subsequence, (i) shows the convergence in $\cH^2_{loc}$ and Proposition~\ref{pr:SMconvPrelocHp} yields (ii).

  (iii)~~This follows from (ii) by using Proposition~\ref{pr:SMconvPrelocHp} twice.
\end{proof}

We return to the semimartingale $R$ of asset returns, which is assumed to be continuous in the sequel. We recall the structure condition~\eqref{eq:StructureContForR} and define $L^\omega(M):=\{\pi\in L(M):\,\|\pi\|_{L^\omega(M)}<\infty\}$, where $\|\pi\|_{L^\omega(M)}:=\|\pi\sint M\|_{\cH^\omega}$ and $\cH^\omega$ was introduced at the end of Section~\ref{se:decomps}.

\begin{Lemma}\label{le:convMandR}
  Let $R$ be continuous, $r\in \{2,\omega\}$, and $\pi,\pi^n\in L^r_{loc}(M)$.
  Then $\pi^n\to\pi$ in $L^r_{loc}(M)$ if and only if $\pi^n\sint R\to \pi\sint R$ in $\cH^r_{loc}$.
\end{Lemma}

\begin{proof}
  By~\eqref{eq:StructureContForR} we have $\pi\sint R=\pi\sint M + \int \pi^\top \,d\br{M}\,\lambda$.
  Let $\chi:=\int \lambda^\top\,d\br{M}\,\lambda$ denote the mean-variance tradeoff process. The inequality
  \[E\Big[\Big(\int_0^T |\pi^\top \,d\br{M}\,\lambda|\Big)^2\Big]
    \leq E\Big[\Big(\int_0^T \pi^\top \,d\br{M}\,\pi\Big) \Big(\int_0^T \lambda^\top \,d\br{M}\,\lambda\Big)\Big]\]
  implies
  $
    \|\pi\sint M\|_{\cH^2}\leq \|\pi\sint R\|_{\cH^2} \leq (1+\|\chi_T\|_{L^\infty})\|\pi\sint M\|_{\cH^2}.
  $
  As $\chi$ is locally bounded due to continuity, this yields the result for $r=2$. The proof for $r=\omega$ is similar.
\end{proof}

\begin{Cor}\label{co:ForWealthProcConv}
  Let $R$ be continuous and $(\pi,\kappa),(\pi^n,\kappa^n)\in \cA$.
  \begin{enumerate}[topsep=3pt, partopsep=0pt, itemsep=1pt,parsep=2pt]
  \item Assume that $\pi^n\to\pi$ in $L^2_{loc}(M)$, that $(\kappa^n)$ is prelocally bounded uniformly in $n$, and   that $\kappa^n_t\to\kappa_t$ $P$-a.s.\ for each $t\in [0,T]$.
      Then $X(\pi^n,\kappa^n)\to X(\pi,\kappa)$ in the semimartingale topology.
  \item Assume $\pi^n\to\pi$ in $L^\omega_{loc}(M)$ and $\kappa^n\to\kappa$ in $\cR^\infty_{loc}$.
      Then $X(\pi^n,\kappa^n)\to X(\pi,\kappa)$ in $\cH^r_{loc}$ for all $r\in[1,\infty)$.
  \end{enumerate}
\end{Cor}

\begin{proof}
  (i)~~By continuity of $\mu$, $\kappa^n_s\sint \mu(ds)_t=\kappa^n_s\sint \mu(ds)_{t_-}$ for all $t$.
  After localization, bounded convergence yields
  $\int_0^T |\kappa^n_t- \kappa_t|\,\mu(dt)\to 0$ $P$-a.s.\ and in $L^2(P)$. Using Lemma~\ref{le:convMandR}, we have $\pi^n\sint R + \kappa^n\sint \mu(dt) \to \pi\sint R + \kappa\sint \mu(dt)$ in $\cH^2_{loc}$.
  In view of~\eqref{eq:wealthExponential} we conclude by Lemma~\ref{le:convStochExp}(ii).

  (ii)~~With Lemma~\ref{le:convMandR} we obtain $\pi^n\sint R + \kappa^n\sint \mu(dt) \to \pi\sint R + \kappa\sint \mu(dt)$
  in $\cH^\omega_{loc}$. Thus the stochastic exponentials converge in $\cH^r_{loc}$ for all $r\in[1,\infty)$ (see
   Theorem~3.4 and Remark~3.7(2) in~\cite{Protter.80}).
\end{proof}

\section{Proof of Lemma~\ref{le:KobylanskiConv3}}\label{se:KobylanskiArgument}
  In this section we give the proof of Lemma~\ref{le:KobylanskiConv3}. As mentioned above, the argument
  is adapted from the Brownian setting of~\cite[Proposition~2.4]{Kobylanski.00}.

  We use the notation introduced before Lemma~\ref{le:KobylanskiConv3}, in particular, recall~\eqref{eq:quadEstimate}.
  We fix $k$ throughout and let $\tau:=\tau_k$. For fixed integers $m\geq n$ we abbreviate
  $\dL=L^n-L^m$, moreover, $\dM$, $\dZ$, $\dN$ have the analogous meaning. Note that $\dL\geq0$ as $m\geq n$.
  The technique consists in applying It\^o's formula to $\Phi(\dL)$, where, with $K:=6 C_k$,
  \[
    \Phi(x)=\frac{1}{8K^2}\big(e^{4Kx}-4Kx-1\big).
  \]
  On $\R_+$ this function satisfies
  \[
    \Phi(0)=\Phi'(0)=0,\q \Phi\geq0,\q \Phi'\geq0,\q \tfrac{1}{2}\Phi''-2K\Phi'\equiv 1.
  \]
  Moreover, $\Phi''\geq0$ and hence $h(x):=\tfrac{1}{2}\Phi''(x)-K\Phi'(x)=1+K\Phi'(x)$
  is nonnegative and nondecreasing.

  (i)~~By It\^o's formula we have
  \begin{align*}
    \Phi(\dL_0)&=\Phi(\dL_\tau)-\int_0^{\tau}\Phi'(\dL_s) \big[f^n(s,L^n_s,Z^n_s)-f^m(s,L^m_s,Z^m_s)\big]\, d\br{M}_s\\
             &\phantom{=} -\int_0^{\tau}\tfrac{1}{2}\Phi''(\dL_s) \,d\,\br{\dM}_s - \int_0^{\tau}\Phi'(\dL_s)\, d\dM_s.
  \end{align*}
  By elementary inequalities we have for all $m$ and $n$ that
  \begin{equation*}%
    |f^n(t,L^n,Z^n)-f^m(t,L^m,Z^m)|^\tau\leq \xi+K\big(|Z^n-Z^m|^2+|Z^n-\tZ|^2+|\tZ|^2\big)^\tau,
  \end{equation*}
  where the index $t$ was omitted. Hence
  \begin{align*}
    \Phi(\dL_0)&\leq\Phi(\dL_\tau)+\int_0^{\tau}\Phi'(\dL_s) \Big[\xi_s + K\big(|\dZ_s|^2 + |Z^n_s-\tZ_s|^2+|\tZ_s|^2\big)\Big]\, d\br{M}_s\\
             &\phantom{=} -\int_0^{\tau}\tfrac{1}{2}\Phi''(\dL_s) \,d\,\br{\dM}_s - \int_0^{\tau}\Phi'(\dL_s)\, d\dM_s.
  \end{align*}
  The expectation of the stochastic integral vanishes since $\dL$ is bounded and $\dM\in\cH^2$.
  We deduce
  \begin{align}
     E\int_0^{\tau} & \big[\tfrac{1}{2}\Phi''(\dL_s)-K\Phi'(\dL_s)\big]|\dZ_s|^2 \, d\br{M}_s
               +E\int_0^{\tau} \tfrac{1}{2}\Phi''(\dL_s) \, d\br{\dN}_s \label{eq:proofKobylanski2a}\\
        & \phantom{\leq}- E\int_0^{\tau}  K\Phi'(\dL_s)|Z^n_s-\tZ_s|^2 \, d\br{M}_s + \Phi(\dL_0)\label{eq:proofKobylanski2b}\\
        & \leq\;\; E\big[\Phi(\dL_\tau)\big]+E\int_0^{\tau}\Phi'(\dL_s) \big[\xi_s + K|\tZ_s|^2\big]\, d\br{M}_s.\label{eq:proofKobylanski2c}
  \end{align}
  We let $m$ tend to infinity, then $\dL_t=L^n_t-L^m_t$ converges to $L^n_t-\tL_t$ $P$-a.s.
  for all $t$ and with a uniform bound, so~\eqref{eq:proofKobylanski2c} converges to
  \[
    E\big[\Phi(L^n_\tau-\tL_\tau)\big]+E\int_0^{\tau}\Phi'(L^n_s-\tL_s) \big[\xi_s + K|\tZ_s|^2\big]\, d\br{M}_s;
  \]
  while~\eqref{eq:proofKobylanski2b} converges to
  \[
    - E\int_0^{\tau}  K\Phi'(L^n_s-\tL_s)|Z^n_s-\tZ_s|^2 \, d\br{M}_s   + \Phi(L^n_0-\tL_0).
  \]
  We turn to~\eqref{eq:proofKobylanski2a}. The
  continuous function $h(x)=\frac{1}{2}\Phi''(x)-K\Phi'(x)$ occurs in
  the first integrand. We recall that $h$ is nonnegative and nondecreasing and note that $\Phi''$ has the same properties.
  Moreover, as $L^m_t$ is monotone decreasing in $m$,
  \[
    h(\dL_s)=h(L^n_s-L^m_s) \uparrow h(L^n_s-\tL_s); \q \Phi''(\dL_s)=\Phi''(L^n_s-L^m_s) \uparrow \Phi''(L^n_s-\tL_s)
  \]
  $P$-a.s. for all $s$. Hence we have for any fixed $m_0\leq m$  that
  \begin{align*}
    E\int_0^{\tau}h(L^n_s-L^m_s)|Z^n_s-Z^m_s|\, d\br{M}_s & \geq E\int_0^{\tau} h(L^n_s-L^{m_0}_s)|Z^n_s-Z^m_s|\, d\br{M}_s;\\
    E\int_0^{\tau}\Phi''(L^n_s-L^m_s)\, d\br{N^n-N^m}_s & \geq E\int_0^{\tau} \Phi''(L^n_s-L^{m_0}_s)\, d\br{N^n-N^m}_s.
  \end{align*}
  The right hand sides are convex lower semicontinuous functions of $Z^m\in L^2(M)$ and $N^m\in \cH^2$, respectively, hence also weakly lower semicontinuous. We conclude from the weak convergences $Z^m\to \tZ$ and $N^m\to \tN$ that
  \begin{align*}
    \liminf_{m\to\infty} E\int_0^{\tau}h(L^n_s-L^m_s)|Z^n_s&-\tZ^m_s|\,  d\br{M}_s \\
      &\geq E\int_0^{\tau} h(L^n_s-L^{m_0}_s)|Z^n_s-\tZ_s|\, d\br{M}_s;
  \end{align*}
  \[
      \liminf_{m\to\infty} E\int_0^{\tau}\Phi''(L^n_s-L^m_s)\, d\br{N^n-N^m}_s  \geq E\int_0^{\tau} \Phi''(L^n_s-L^{m_0}_s)\, d\br{N^n-\tN}_s
  \]
  for all $m_0$. We can now let $m_0$ tend to infinity, then by monotone convergence the first right hand side tends to
  $
    E\int_0^{\tau} h(L^n_s-\tL_s) |Z^n_s-\tZ_s|\, d\br{M}_s
  $
  and the second one tends to
  \[
    E\int_0^{\tau} \Phi''(L^n_s-\tL_s)\, d\br{N^n-\tN}_s\geq 2 E\int_0^{\tau}  d\br{N^n-\tN}_s=2 E\big[\br{N^n-\tN}_{\tau}\big],
  \]
  where we have used that $L^n-\tL\geq 0$ and $\Phi''(x)=2 e^{4Kx}\geq 2$ for $x\geq0$.
  Altogether, we have passed from~\eqref{eq:proofKobylanski2a}--\eqref{eq:proofKobylanski2c} to
  \begin{align*}
     E&\int_0^{\tau} \Big(\tfrac{1}{2}\Phi''-2K\Phi'\Big)(L^n_s-\tL_s)\,|Z^n_s-\tZ_s|^2 \, d\br{M}_s
     +  E\big[\br{N^n-\tN}_{\tau}\big] \\
        & \leq E\Phi(L^n_\tau-\tL_\tau)-\Phi(L^n_0-\tL_0)+ E\int_0^{\tau}\Phi'(L^n_s-\tL_s) \big[\xi_s + K|\tZ_s|^2\big]\, d\br{M}_s.
  \end{align*}
  As $\tfrac{1}{2}\Phi''-2K\Phi'\equiv 1$, the first integral reduces to $E\int_0^{\tau} |Z^n_s-Z_s|^2 \, d\br{M}_s$.
  If we let $n$ tend to infinity, the right hand side converges to zero by dominated
  convergence, so that we conclude
  \[
   E\int_0^{\tau} |Z^n_s-\tZ_s|^2 \, d\br{M}_s\to 0; \q E\big[\br{N^n-\tN}_{\tau}\big]\to 0
  \]
  as claimed.

  (ii)~~For all $m$ and $n$ we have
  \begin{align}\label{eq:proofKobylanski3}
    |L^n_{t\wedge\tau}-L^m_{t\wedge\tau}|&\leq |L^n_\tau-L^m_\tau|+ \int_{t\wedge\tau}^{\tau}| f^n(s,L^n_s,Z^n_s)-f^m(s,L^m_s,Z^m_s)|\,d\br{M}_s\nonumber\\
                 &\phantom{\leq}+ \big| (M^n_\tau-M^m_\tau) - (M^n_{t\wedge\tau}-M^m_{t\wedge\tau})\big|.
  \end{align}
  The sequence $M^m=Z^m\sint M + N^m$ is Cauchy in $\cH^2_{\tau}$.
  We pick a fast subsequence, still denoted by $M^m$, such that
   $\|M^{m}-M^{m+1}\|_{\cH^2_\tau}\leq 2^{-m}$. This implies that
   \[
     M^*:=\sup_m |M^m| \in \cH^2_\tau;\q Z^*:=\sup_m |Z^m|\in L^2_\tau(M)
   \]
   and that $Z^m$ converges $P\otimes \br{M^\tau}$-a.e.~to $\tZ$.
   Therefore,
  $\lim_n f^m(t,L^m_t,Z^m_t)=f(t,\tL_t,\tZ_t)$ $P\otimes \br{M^\tau}$-a.e.
  Moreover, $|f^m(t,L^m_t,Z^m_t)^\tau|\leq \xi_t +C|Z^*_t|^2$ and this bound is in $L^1_\tau(M)$.
  Passing to a subsequence if necessary, we have
  \begin{align*}
     \lim_{m\to\infty}  \int_{0}^{\tau}| & f^n(s,L^n_s,Z^n_s)-f^m(s,L^m_s,Z^m_s)|\,d\br{M}_s  \\
      &=\int_{0}^{\tau}| f^n(s,L^n_s,Z^n_s)-f(s,\tL_s,\tZ_s)|\,d\br{M}_s\q P\mbox{-a.s.}
  \end{align*}

  As $M^m\to \tM$ in $\cH^2_\tau$, we have $E\big[\sup_{t\leq T} |M_{t\wedge\tau}^m-\tM_{t\wedge\tau}| \big]\to 0$ and, after picking a subsequence,
  $\sup_{t\leq T} |M_{t\wedge\tau}^m-\tM_{t\wedge\tau}|\to 0$ $P$-a.s.
  We can now take $m\to\infty$ in~\eqref{eq:proofKobylanski3} to obtain
  \begin{align*}
    \sup_{t\leq T}|L^n_{t\wedge\tau}-\tL_{t\wedge\tau}|&\leq |L^n_\tau-\tL_\tau|+ \int_{0}^{\tau}| f^n(s,L^n_s,Z^n_s)-f(s,\tL_s,\tL_s)|\,d\br{M}_s\\
                 &\phantom{\leq}+ \sup_{t\leq T}\big| (M^n_\tau-\tM_\tau) - (M^n_{t\wedge\tau}-\tM_{t\wedge\tau})\big|.
  \end{align*}
  With exactly the same arguments, extracting another subsequence if necessary, the right hand
  side converges to zero $P$-a.s.\ as $n\to\infty$. We have shown that
  $\lim_{n} \sup_{t\leq T}|L^n_{t\wedge\tau}-\tL_{t\wedge\tau}|=0$, along a subsequence. But by monotonicity, we conclude the result for the whole sequence.\hfill$\square$

\bibliography{C:/Users/numadmin/Documents/tex/stochfin}
\bibliographystyle{plain}

\end{document}